\documentclass{article}
\usepackage{amsmath, amssymb}
\usepackage{amsthm}

\usepackage{graphicx}
\usepackage{geometry}
\usepackage{multirow}

\usepackage{bm}
\usepackage{graphicx, graphics}
\usepackage{epsfig}
\usepackage{xcolor}

\usepackage{natbib}

\newtheorem{prop}{\protect\propositionname}
\newtheorem{rem}{\protect\remarkname}
\newtheorem{thm}{\protect\theoremname}
\newtheorem{lem}{\protect\lemmaname}

\providecommand{\lemmaname}{Lemma}
\providecommand{\propositionname}{Proposition}
\providecommand{\remarkname}{Remark}
\providecommand{\theoremname}{Theorem}

\newcommand{\comment}[1]{\textcolor{blue}{\it #1}}
\newcommand{\noprint}[1]{}

\newcommand{\petit}{\hspace{-0.1cm}}

\title{Mean curvature and mean shape for multivariate functional data under Frenet-Serret framework}

\author{Juhyun Park and Nicolas J-B. Brunel \\
Lancaster University and ENSIIE \\
email: \href{mailto:juhyun.park@lancaster.ac.uk}{juhyun.park@lancaster.ac.uk}\,, \href{mailto:nicolas.brunel@ensiie.fr}{nicolas.brunel@ensiie.fr}}

\date{September 12, 2019}

\begin{document}
\maketitle

\bibliographystyle{chicago}

\begin{abstract}
The analysis of curves has been routinely dealt with using tools from functional data analysis. 
However its extension to multi-dimensional curves poses a new challenge due to its inherent geometric features that are difficult to capture with the classical approaches that rely on linear approximations.
We propose a new framework for functional data as multidimensional curves that allows us to extract geometrical features from noisy data. We define a mean through measuring shape variation of the curves. The notion of shape has been used in functional data analysis somewhat intuitively to find a common pattern in one dimensional curves.
As a generalization, we directly utilize a geometric representation of the curves through the Frenet-Serret ordinary differential equations and introduce a new definition of mean curvature and mean shape through the mean ordinary differential equation. We formulate the estimation problem in a penalized regression and develop an efficient algorithm. We demonstrate our approach with both simulated data and a real data example. 
\end{abstract}

\section{Introduction }

We consider the problem of analyzing a set of three-dimensional curves in $\mathbb{R}^3$, as an instance of multivariate functional data. 
A typical example would be the recordings of spatial coordinates for tracking movements of body parts or objects (e.g.,  \citet{RaketMarkussen2016, FlashHogan1985}).
A standard assumption with functional data analysis (FDA) \citep{RamsaySilverman2005, FerratyVieu2006, WangMuller2016} is that there exists a common structure, often through a common mean and variance function. Then variability decomposition around the mean through functional principal component analysis provides a parsimonious decomposition of the variability. This type of linear approximations is powerful as it allows us to naturally extend tools from scalar curves to multidimensional curves \citep{ChiouChenYang2014}. Nevertheless, such analytic extension can also hide some important features in multidimensional curves. This is illustrated in Figure~\ref{fig1}, where we show the original three-dimensional curve data, together with the corresponding three marginal one-dimensional curves. 
The commonality that can be extracted from the three-dimensional plot on the right panel is difficult, even by eye, to match to the information obtained from the left panel, although there does seem an {\it obvious} common structure in the marginal curves. 
The problem becomes more complex with larger dimensional curves, where we cannot even easily visualize the data. 
\begin{figure}[htbp]
\begin{center}
\begin{tabular}{cc}
\includegraphics[width=0.48\textwidth]{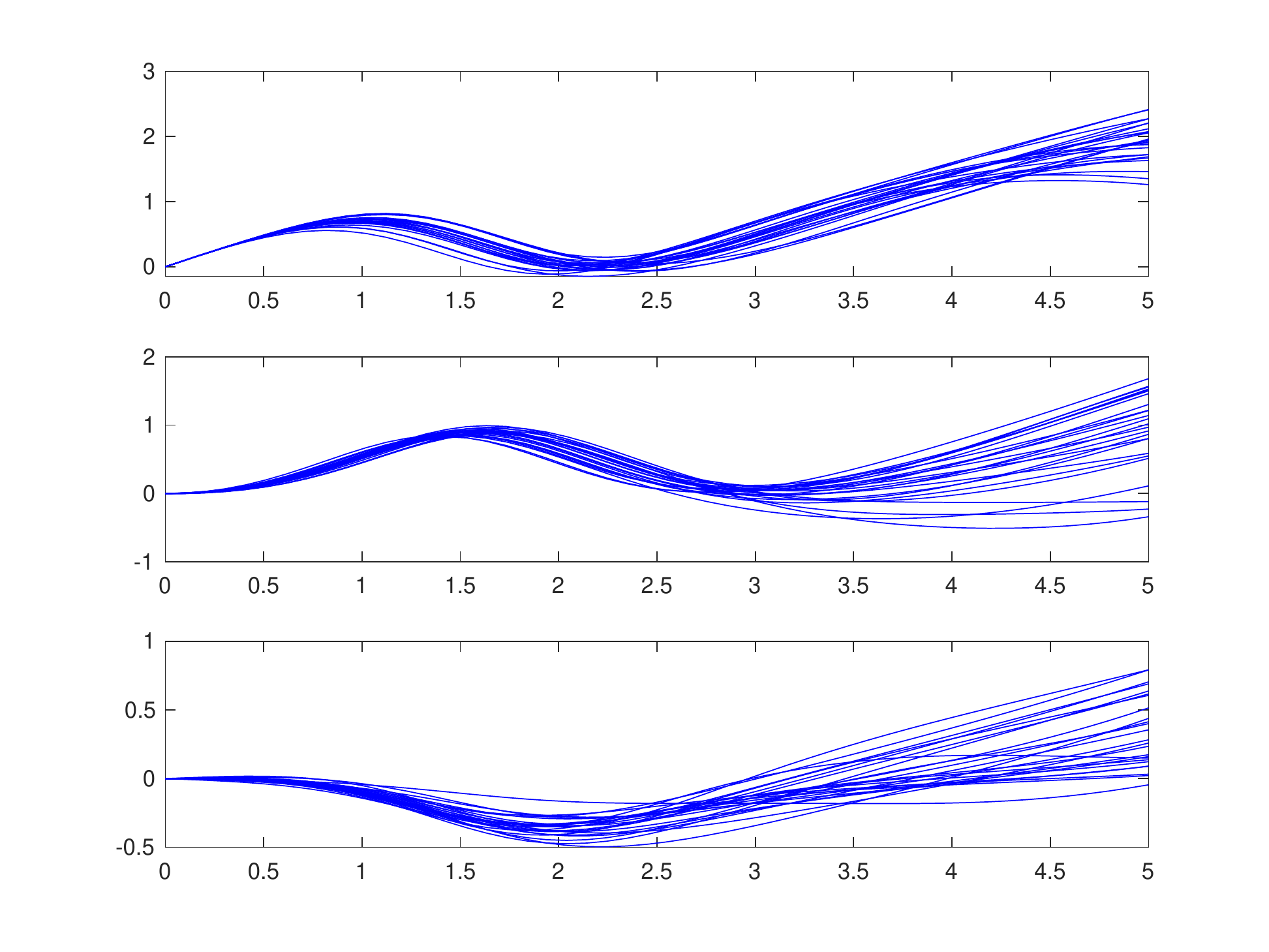} &
\hspace*{-0.1cm} \includegraphics[width=0.48\textwidth]{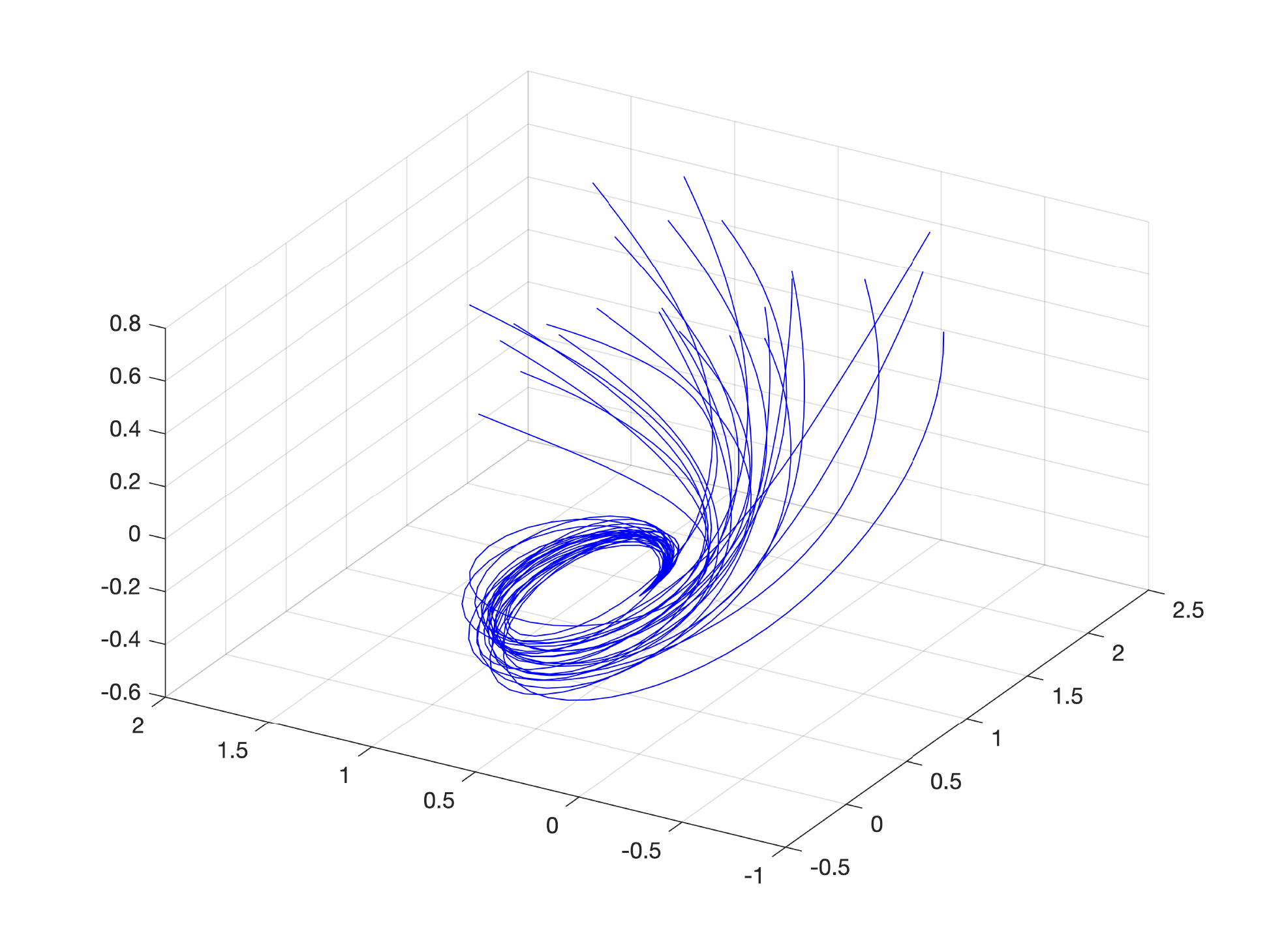}
\end{tabular}
\end{center}
\caption{\label{fig1}Three-dimensional curves in marginal coordinates plot (left) and 3-d plot (right), population of $N=25$ space curves simulated as in section 5.2.}
\end{figure}

An alternative approach to analysing three-dimensional curves has been through estimating curvature and torsion \citep{Kim2013, Sangalli2009, Lewiner2005}. As these parameters have a physical meaning, such analysis could offer more interpretable summaries. However, the focus has been more on estimating derivatives in nonparametric regression and the link to functional data is somewhat lost. 
We propose a new statistical framework to formulate this type of geometric analysis of functional data, borrowing ideas from statistical shape analysis \citep{Dryden1998, Younes2010}.

The geometric notion of shape has been used in FDA somewhat intuitively in defining {\it structural mean} in the presence of phase variation \citep{KneipGasser1992}, more generally with manifold structure \citep{ChenMuller2012}, and in detecting outliers in multivariate functional data \citep{DaiGenton2018}. 
A general definition of shape is given as what is left invariant under the actions of the rigid transformations of the Euclidean space, i.e., rescaling, translating and rotating. The definition alone however is too general to be useful in statistical modelling considered in FDA. 
Shapes are formally considered as equivalence classes under some appropriate group actions, and the classical statistical methodologies need to be adapted in order to deal with the non-Euclidean properties of shape spaces. Several ways of constructing shape spaces (or feature spaces) have been proposed: discretization with landmarks based after discretization \citep{Dryden1998}, or infinite dimensional shape spaces \citep{Younes2010, Srivastava2011}. Such shape or geometric statistical analysis requires a definition of distance (inducing a Riemannian structure for instance), and the natural extension of the usual mean is defined as a Fr\'echet mean. 
Some of these ideas, such as the elastic shape analysis has been suggested for analyzing the variations of functional data, typically for the registration problem in one-dimensional curves \citep{kurtek2012statistical}. Its extension to multi-dimensional curves is found in \citet{SrivastavaKlassen2016}. 

A natural strategy in shape analysis would be to define a proper space equipped with a good metrics where one can define a Fr\'echet mean. On the contrary, we consider a new distance between curves that does not depend on the usual Cartesian coordinate system but uses a parameterization of the space of smooth curves based on the Frenet-Serret representation, introduced in Section~\ref{sec:FDA-frenet}. We directly exploit the Frenet paths that are interpreted as a representative of the equivalence class {\it shape}. Within this framework, we introduce a new definition of mean shape through the mean Ordinary Differential Equation (ODE) (or flow). 
The new definition of mean shape requires a proper estimation of functional parameters such as curvature as well as Frenet frames.
In particular, we introduce the notion of mean curvature within this framework and show that the estimation of curvature and Frenet frames can be cast into the problem of an ODE estimation in Lie group.
To accompany the new definition of the mean, we propose a statistical framework for inference and develop an efficient algorithm. 
In general, the ODE estimation is a difficult problem (e.g., \cite{Ramsay2007}), especially involving nonparametric estimation of time-varying parameter (e.g. \cite{Mueller2010, WuDing2014}), even  without the orthogonality constraint required in our formulation. As a by-product, our formulation also offers a new solution to a non-trivial ODE inference problem. We refer the readers to \citet{RamsayHooker2017} for recent development on data analysis with ODE models. 

The paper is organized as follows. Section~\ref{sec:FDA-frenet} introduces the geometrical representation of the curves based on the Frenet-Serret framework. Section~\ref{sec:FDA-ShapeAnalysis} defines a mean shape and studies properties related to statistical shape analysis. Section~\ref{sec:estimation} gives a statistical framework for estimation, followed by numerical studies in Section~\ref{sec:numeric}. Main proofs to the assertions in the article are given in Supplementary Material.

\section{A geometric curve representation and Frenet-Serret Framework}\label{sec:FDA-frenet}

We consider a set of curves in $\mathbb{R}^3$ defined as functions $\{x: [0, T] \rightarrow  \mathbb{R}^3\}$. 
In order to avoid some technical difficulties, we assume that the curves are \emph{regular, i.e} they are of class $C^{r}, r\geq 3$ (w.r.t time $t$), the time derivative $\dot{x}(t)$ never vanishes on $[0,T]$ and the curves never intersect themselves \citep{kuhnel2015differential}. 

For a curve $x$, define the arclength $s(t)=\int_{0}^{t}\Vert \dot{x}(u)\Vert _{2}\,du, t \in [0, T]$,
where $\| \cdot\| _{2}$ is the Euclidean norm in $\mathbb{R}^3$
and $s(T)=L$ is the total length of the curve $X=\left\{ x(t),\,t\in\left[0,T\right]\right\}$.
The \emph{shape} of the curve $X: \:\left[0,L\right]\longrightarrow\mathbb{R}^{3}$ is the image of the function $x$, which satisfies $x(t)=X(s(t))$.
The derivation with respect to arclength $s$ is denoted with prime
i.e $Y'(s)=\frac{d}{ds}Y(s)$, whereas time differentiation
is always denoted by a dot. 

\subsection{Time warping and shape variation of the curves}

The above arclength parametrization should not be confused with the standard representation of time warping or phase variation in the functional data. Suppose that $x_i$ is given as $x_i(t) = x(h_i(t))$, where $h_i$ are warping functions. As $\dot{x_i}(t) = {\dot x}(h_i(t))\dot h_i(t)$, by change of variables, the corresponding arclength can be expressed as
\[
s_i(t) = \int_0^t \|{\dot x}(h_i(t)){\dot h_i}(t)\|_2 \,dt = \int_0^{h_i(t)}\|{\dot x}(u)\|\,du = s(h_i(t)) \,.
\]
It follows that  $x_i(t) = x(h_i(t)) = X(s(h_i(t)) = X(s_i(t))$, that is, the \emph{shape} of the curve is preserved under time warping. For univariate functional data, phase variation expressed as time warping functions is often confounded with shape variation as we introduce here.

In this work, we distinguish between phase ($h_i$) and shape variation (${X}_i$), and explicitly model the shape variation in the spirit of functional data analysis, based on its geometric curve representation.  

\subsection{Frenet-Serret equation}\label{sec:Frenet}

The arclength parametrization of the curve $x(t) = X(s(t))$ implies that $\dot{x}(t)=\dot{s}(t)X'(s(t))$, meaning that the tangent vector $T(s)\triangleq {X}'(s)$ is unit length for all $s$ in $\left[0,L\right]$. 
The curvature can be defined as $s\mapsto\kappa(s)=\| T'(s)\| $, and if $\kappa>0$, a moving frame is defined with the addition of Normal vector $N(s)=\frac{1}{\kappa(s)}T'(s)$, and the bi-normal vector $B(s)=T(s)\times N(s)$ to the tangent vector $T$. The torsion
$s\mapsto\tau(s)$ is the function that satisfies $B'(s)=-\tau(s)N(s)$
for all $s$ in $\left[0,L\right]$. The curvature and torsion
are geometric invariant of the curve, independent of the
parametrization of a curve $X$. They are also invariant
under the action of rigid (Euclidean) motions. These functional parameters can be directly defined with extrinsic formulas as 
\begin{equation}
\kappa(s(t))=  \frac{\Vert \dot{x}(t)\times\ddot{x}(t)\Vert _{2}}{\Vert \dot{x}(t)\Vert _{2}^{3}}\,,\qquad
\tau(s(t))=  \frac{\langle \dot{x}(t)\times\ddot{x}(t),\dddot{x}(t)\rangle }{\Vert \dot{x}(t)\Vert _{2}^{3}} \,.
\label{eq:ExtrinsicCurvatures}
\end{equation}
Although the formulas are useful for computing curvature and torsion in practice, the geometrical interpretation of these parameters is somewhat hidden in these expressions.
Indeed, the vectors $s\mapsto T(s),\,N(s),\,B(s)$ are tightly related
through Frenet-Serret ODE
\begin{equation}
\left\{ \begin{array}{ll}
T'(s)= & \kappa(s)N(s)\\
N'(s)= & -\kappa(s)T(s)+\tau(s)B(s)\\
B'(s)= & -\tau(s)N(s)
\end{array}\right.\label{eq:FS-ode-vector}
\end{equation}
with an initial condition $T(0),N(0),B(0)$. In other words, the moving frame defines
a curve $s\mapsto Q(s)=\left[T(s)\vert N(s)\vert B(s)\right]$ in
the group of special orthogonal matrices $SO(3)$. As $SO(3)$ is
a Lie group with a manifold structure, the Frenet-Serret ODE can be seen as an ODE defined in the Lie
group with
\begin{eqnarray}
\dot{Q}(s) & = & Q(s)A(s)\label{eq:FS-ode-matrix}
\end{eqnarray}
where 
\begin{equation}
A(s)=\left[\begin{array}{ccc}
0 & -\kappa(s) & 0\\
\kappa(s) & 0 & -\tau(s)\\
0 & \tau(s) & 0
\end{array}\right].\label{eq:matrixA}
\end{equation}
We shall denote by $\theta$ the functional parameters $(\kappa, \tau)$ with the set of admissible parameters by $\mathcal{H} = \{\theta = (\kappa,\tau), \kappa>0,   \kappa, \tau \in C^2 \}$, and by $A_\theta$ the corresponding skew-symmetric matrix. 

For regular curves in $\mathbb{R}^{p}$ for $p>3$, the same moving
frame in $SO(p)$ can be defined, in terms of a skew-symmetric matrix
similar to (\ref{eq:matrixA}), and the so-called generalized
curvatures $\kappa_{1},\kappa_{2},\dots,\kappa_{p-1}$ \citep{kuhnel2015differential}. 

\subsection{Effect of scaling}

For invariance with respect to scaling,  we define the equation with the scaled curves to the unit length.
Rescaling does change the geometry only through a scaling factor, i.e the matrix $s\mapsto A(s)$ in the ODE (\ref{eq:FS-ode-matrix}) is also renormalized and the rescaled curves
$\frac{1}{L}{X}(s)$ have new curvilinear arclength
$\tilde s=s/L$ and the rescaled Frenet paths are $\tilde{Q}(\tilde s)=Q(\tilde sL)$.
The rescaled Frenet-Serret ODE, defined on $\left[0,1\right]$ is
$\tilde{Q}^{'}(\tilde s)=\tilde{Q}(\tilde s)\tilde{A}(\tilde s)$ with $\tilde{A}(\tilde s)=LA(\tilde sL)$,  implying that rescaling a curve by $1/L$, multiplies its curvature and torsion by $L$. 

\subsection{Frenet equation and its equivalent class}\label{sec:FrenetSerretODE-and-Equivalence}

We call the solutions $s\mapsto Q(s)$ of the Frenet-Serret equations the Frenet
paths, and the set of Frenet paths is denoted by 
\[
\mathcal{F}=\left\{ s\mapsto Q_{\theta}(s)\vert Q^{'}(s)=Q(s)A_{\theta}(s),\:s\in\left[0,1\right],\,Q(0)\in SO(3),\,\theta\in\mathcal{H}\right\} \,.
\]
Among the set of all Frenet paths, we pay a particular attention to 
the subset of Frenet paths with initial condition equal to the
identity matrix $I_{p}$, $\mathcal{F}_{0}=\left\{ Q\in\mathcal{F}\vert Q(0)=I_{3}\in SO(3)\right\}$.
In short, we deal with the set of arclength-parametrized regular
curves of length 1, denoted by $\mathcal{C}_{1}$. As any regular curve
$X$ can be recovered by integrating its tangent $\dot{X}(s)=T_{\theta}(s)$,
we have 
\begin{equation}
\mathcal{C}_{1}=\left\{ s\mapsto {X}(s)=X_{0}+Q_{0}\int_{0}^{s}T_{\theta}(u)du\:\vert\:Q_{\theta}\in\mathcal{F}_{0},\,X_{0}\in\mathbb{R}^{3},\,Q_{0}\in SO(3)\right\} \label{eq:Space_Regular}
\end{equation}
indexed by the parameters $\left(X_{0},Q_{0},\theta\right)$.
This parametrization is known to be one-to-one: for any curve $X_{1}$ and $X_{2}$ having the same curvature $\theta$, there exists a vector $a$ and a rotation $R\in SO(p)$ such that ${X}_{1}=a+R{X}_{2}$. 
If the Frenet path for ${X}_{2}$ has an initial condition
equal to $I_{3}$, the rotation matrix $R$ is exactly the initial
condition of the Frenet path associated with ${X}_{1}$. For
this reason, the space $\mathcal{F}_{0}$ can be naturally considered
as the shape space. 
The functions $s\mapsto\theta(s)$ or $s\mapsto Q_{\theta}(s)$ represent the geometrical content of any regular curve ${X}$.

\section{Functional Data Analysis and Statistical Shape Analysis\label{sec:FDA-ShapeAnalysis}}

Let us consider $N$ curves in $\mathbb{R}^{3}$ defined as functions $\mathcal{S}=\left\{ x_{1},\dots,x_{N}\right\}$ from $\left[0,T\right]$ to $\mathbb{R}^{3}$. 
As seen in section \ref{sec:FrenetSerretODE-and-Equivalence}, we can identify the shapes with the Frenet paths $\boldsymbol{Q}=\left\{ Q_{1},\dots,Q_{N}\right\} $, or equivalently with the curvatures $\boldsymbol{\theta}=\left\{ \theta_{1},\dots,\theta_{N}\right\} $. 
Hence, our idea is to define a mean shape for the population of curves $\mathcal{S}$, through a reference
pattern for curves in $\mathcal{C}_{1}$, independently of the variations
in translations, rotations and scalings. Our parametrization (\ref{eq:Space_Regular})
shows that the quotient space of arclength parametrized curves (under
the group action of Euclidean motions) is exactly the space of Frenet
paths. Our aim is then to derive a reference parameter $\bar{\theta}$
(and reference Frenet path $\bar{Q}$) for $\mathcal{S}$. 

In order to define a mean shape based on the Frenet-Serret representation, we first study the characteristic features of Frenet Paths. We introduce some important concepts in view of statistical shape analysis that provide the basis of our formulation.   
In particular, we provide a link between the skew-symmetric matrix $s\mapsto A_{\theta}(s)$ (or the curvature $\theta$) and the Frenet path $s\mapsto Q(s)$ defined by the corresponding Frenet-Serret ODE.

\subsection{Differential equations on Lie Groups}

It is convenient to consider the Frenet-Serret ODE as an Ordinary Differential Equations on the Lie group $SO(p)$. 
Ensuring the orthogonality constraint requires a special treatment in developing a numerical algorithm to solve an ODE and also in tackling a parameter estimation problem in ODE, as numerical errors can accumulate and induce an uncontrolled bias. 
The extension of the theory of ODEs from Euclidean space to Lie groups or manifolds is well developed \citep{Hairer2006}. 
In particular, the rotation group $SO(p)$ is a Lie Group that is also a differentiable manifold, with many remarkable properties that are essential in tackling the numerical problems \citep{Absil2010}.

Typically, $SO(p)$ is considered as a submanifold of the Euclidean space $\mathbb{R}^{p\times p}$, with the usual inner product $\left\langle M,M'\right\rangle =\mbox{Tr}\left(M^{\top}M'\right)$
(and the associated Frobenius norm).
The Tangent Space at point $M$ to $SO(p)$ is the vector space 
\[
T_{M}SO(p)=\left\{ MU\vert U^{\top}=-U\right\}\,,
\]
usually identified with the set of skew-symmetric matrices. In
particular, the Tangent Space at the identity $I_{p}$ is called the
Lie algebra of the Lie group, denoted by $\mathfrak{so}(p)$.
An ODE is then defined as a function $F\::\:SO(p)\longrightarrow T_{M}SO(p)$,
such that $\dot{Y}(t)=F\left(t,Y(t)\right)$. In the case of the Frenet-Serret
equation, the vector field is time-varying but relatively simple.

A fundamental tool for the analysis of ODE and flows on Lie groups
is the Exponential map, $\mbox{Exp}_{M}$, at point $M$, which relates
the tangent space to the manifold. The Exponential map $\mbox{Exp}_{M}\,:\,T_{M}SO(p)\longrightarrow SO(p)$
is such that $\mbox{Exp}_{M}(U)=\gamma(1;M,U)$, where $\gamma$ is
the unique geodesic $s\mapsto\gamma(s,M,U)$ such that $\gamma(0;M,U)=M$
and $\dot{\gamma}(0;M,U)=U$. Conversely, if we have a given root
$M$ and a target point $N$, the logarithmic map returns a tangent
vector at $M$, pointing toward $N$, of length $dist(M,N)$. Hence,
the logarithmic map $\mbox{Log}_{M}\,:\,SO(p)\longrightarrow T_{M}SO(p)$
at $M$ is $\mbox{Log}_{M}(N)=V$ such that $\mbox{Exp}_{M}(V)=N$
and $\| \mbox{Log}_{M}(N)\| =dist(M,N)$. 

Interestingly, if we consider a matrix Lie group, the exponential
and logarithmic maps can be expressed simply with the classical matrix
exponential $\exp(A)=\sum_{k\geq0}\frac{A^{k}}{k!}$ and matrix logarithm,
see \citet{Higham2008}. Indeed, we have $\mbox{Exp}_{M}(U)=M\exp(U)$
and $\mbox{Log}_{M}(N)=\log(M^{\top}N)$. As a consequence, the geodesic
distance $dist(M,N)=\| \log(M^{\top}N)\| _{F}$ has
a closed form expression that is amenable to computation. Numerous efficient algorithms exist for computing the exponential of a matrix; the case of $p=3$ is remarkable, as in that case the exponential and logarithm have a closed-form expression given by the so-called Rodrigues formula. We will use in our applications these formulas to derive our fast algorithms. 


A fruitful approach to solving a differential equation $\dot{Y}=A(t)Y$
with $Y(0)=Y_{0}$ in a Lie group is to look for a solution of the
form $Y(t)=Y_{0}\exp\left(\Omega(t)\right)$. This implies that the
function $t\mapsto\Omega(t)$ is defined in $\mathfrak{so}(p)$, and
is itself a solution of the \emph{dexpinv} differential equation (chapter
IV.7 in \citet{Hairer2006}, \citet{Iserles2000}). Consequently, the function admits the so-called Magnus expansion {\footnotesize\begin{equation}
\hspace*{-0.1cm}\Omega(t)=\int_{0}^{t}\petit A(s)ds-\frac{1}{2}\int_{0}^{t}\petit\left[\int_{0}^{\tau}\petit A(s)ds,A(\tau)\right]\petit d\tau
+\frac{1}{4}\int_{0}^{t}\petit\left[\int_{0}^{\tau}\petit \left[\int_{0}^{\sigma}\petit A(\mu)d\mu,A(\sigma)\right]\petit d\sigma,A(\tau)\right]\petit d\tau+\dots\,,\label{eq:MagnusExpansion}
\end{equation}}
which can be used to derive efficient integration methods. 
Utilizing the Magnus expansion, we obtain an exponential form for the flow of
the ODE (\ref{eq:FS-ode-matrix}). The flow $\phi_{\theta}$ is the
function such that for all $s$ in $\left[0,1\right]$, $t\mapsto\phi_{\theta}\left(t,s,M\right)$
is the solution of the Frenet-Serret ODE satisfying $Q(s)=M$. For
all $t,s$ in $\left[0,1\right]$ and $\theta$ in $\mathcal{H}$,
we define the matrix-valued (in $\mathfrak{so}(p)$) function $\Omega\left(t,s;\theta\right)$
such that the flow can be written as 
\begin{equation}
\phi_{\theta}\,:\,\left(t,s,Q\right)\mapsto Q\exp\left(\Omega(t,s;\theta)\right)\label{eq:ExactODEFlow}
\end{equation}
The essential property of the flow is the group property, i.e for
all $s,u,t\in\left[0,1\right]$ and $Q\in SO(p)$, $\phi_{\theta}\left(t,s,Q\right)=\phi_{\theta}\left(t-u,u,\phi_{\theta}\left(u-s,s,Q\right)\right)$. 

Although the expression of $\Omega$ is intractable in general, we
can derive a consistent approximations to the flow by truncating the
Magnus expansion, see chapter IV in \citet{Hairer2006}. In particular,
we use an approximation of order 2, obtained by using a simple quadrature
rule with the midpoint and truncating after the first term: $Q_{s+h}=Q_{s}\exp\left(hA_{\theta}\left(s+\frac{h}{2}\right)\right)$,
i.e. $\phi_{\theta}(h,s,Q_{s})-Q_{s+h}=O(h^{2})$. The corresponding
approximate flow $\tilde{\phi}_{\theta}\left(h,s,Q\right)=QN_{h}(s,\theta)$
can be seen as an Euler-Lie method that possesses several interesting
features: it respects the $SO(p)$ constraint, has an explicit
and pointwise dependence in $\theta$, and the approximation is uniform
on $SO(p)$.

In the computation and in the analysis of approximation error of our
algorithms, we will use repeatedly the commutator $\left[A,B\right]$
between $A$ and $B$ in the Lie algebra $\mathfrak{so}(p)$ defined
as $\left[A,B\right]=AB-BA$. More generally, the commutator between
two vector fields measures and computes the degree of non-commutativity between two vector fields. In the case of matrix Lie groups, it boils down to the classical matrix commutator $\left[A,B\right]$,
often denoted by $\mbox{ad}_{A}(B)$ (derivative of the
adjoint representation). The commutator arises in the Baker-Campbell-Hausdorff (BCH)
formula, which is central in the theoretical and computational analysis
of functions on Lie Groups as 
\begin{equation}
\exp(tA)\exp(tB)=\exp\left(tA+tB+\frac{1}{2}t^{2}\left[A,B\right]+O(t^{2})\right)\label{eq:BCHformula}
\end{equation}
for $t$ small enough.

\subsection{Mean shape and mean vector field \label{sub:Mean-shape-ODE}}

As seen in section \ref{sec:FrenetSerretODE-and-Equivalence}, we identify the shapes with the Frenet paths $\boldsymbol{Q}=\left\{ Q_{1},\dots,Q_{N}\right\} $, or equivalently with the curvatures $\boldsymbol{\theta}=\left\{ \theta_{1},\dots,\theta_{N}\right\} $. 
Furthermore, we have shown that the geometrical features $\theta_{i}$ define a vector field, and that the observable features such as tangent, normal or binormal vectors are in fact the
corresponding flows. These observations lead us to considering the mean shape
as the mean of the vector fields $Q\mapsto QA_{\theta_{i}}(s)$ with
flows $\phi_{\theta_{i}}$ defined by (\ref{eq:ExactODEFlow}). We
define then the mean vector field (or mean shape) as the vector field
defined on $SO(3)$ such that the solution paths are close to the
individual Frenet paths $Q_{i},\:i=1,\dots,N$. In other words, the
mean vector field corresponds to the flow that provides {\it a best} approximation
to all the individual flows. 

A noticeable feature of our formulation is that we do not use the infinitesimal characterization of the differential equation based on the derivative. We use instead the group property of the flow that can be interpreted as a \emph{self-prediction property}: if $s\mapsto Q(s)$ is a solution to equation (\ref{eq:FS-ode-matrix}), then for all $t,\,s\in\left[0,1\right]$ such that $\left|t-s\right|\leq1$, we have 
\begin{equation}
Q(t)=\phi_{\theta}\left(t-s,s,Q(s)\right).\label{eq:selfprediction}
\end{equation}
Otherwise, the curve $s\mapsto Q(s)$ is a solution to $Q'=QA_{\theta}$
if and only if $$\int_{0}^{1}\int_{0}^{1}d\left(Q(t),\phi_{\theta}\left(t-s,s,Q(s)\right)\right)dsdt=0,$$
where $d(\cdot,\cdot)$ is a distance defined on $SO(p)$. In the
case of multiple solutions $Q_{1},\dots,Q_{N}$, the mean flow $Q\mapsto Q\exp\left(\Omega(t-s,s;\theta)\right)$
should minimize the self-prediction errors for all the trajectories
simultaneously. In this work, the prediction error is measured with
the geodesic distance in $SO(3)$. 

For statistical estimation, we first introduce the penalized criterion $\breve{\mathcal{I}}_{h,\lambda}(\theta;\boldsymbol{Q})$ defined as 
{\footnotesize\begin{equation}
\frac{1}{N}\sum_{i=1}^{N}\int_{0}^{1}\int_{0}^{1}K_{h}(t-s)\| \log\left(Q_{i}(t)^{\top}Q_{i}(s)\exp\left(\Omega(t-s,s;\theta)\right)\right)\| _{F}^{2}dsdt+\lambda\int_{0}^{1}\| \theta^{''}(t)\| ^{2}dt \label{eq:SelfPrediction1}
\end{equation}}

\noindent where $K(\cdot)$ is a kernel function with compact support,
e.g. $K(u)=\frac{3}{4}(1-u)^{2}1_{[-1,1]}(u)$. As usual, we denote
the scaled kernel by $K_{h}(u)=(1/h)K(u/h)$, and the absolute
moments of $K$ are denoted by $\mu_m(K)=\int_{-1}^{1}\left|x\right|^{m}K(x)dx$. 
The kernel $K(\cdot)$ and the bandwidth $h$ define a prediction
horizon for the flow. For a fixed $h$ and $\lambda$, we can define
the mean vector field (or curvature) as the parameter $\theta$ that
minimizes the global prediction error $\breve{\mathcal{I}}_{h,\lambda}(\theta;\boldsymbol{Q})$.
In the next proposition, we show that the mean vector field $\breve{\theta}_{h,\lambda}$ exists under general conditions.
\begin{prop}
\label{ExistenceUniquenessMean} Let $Q_{1},\dots,Q_{N}$ be Frenet
paths with curvatures $\boldsymbol{\theta}$, such that for all $i=1,\dots,N$,
$\| \theta_{i}\| _{\infty}\leq C$. There exists
$\breve{\theta}_{h,\lambda}$ in $\mathcal{H}$ such that 
\[
\breve{\theta}_{h,\lambda}\in\arg\min_{\theta\in\mathcal{H}}\breve{\mathcal{I}}_{h,\lambda}(\theta)\,.
\]
\end{prop}
Proposition \ref{ExistenceUniquenessMean} shows that the mean vector
field exists for any $h$ and $\lambda$ in great generality, 
as long as the sample is bounded in $L^{2}$. We can also define the
mean Frenet Path $\breve{Q}_{h,\lambda}(t)=\exp\big(\Omega(t,0,\breve{\theta}_{h,\lambda})\big)$ and the corresponding mean shape $\breve{{X}}$ obtained by integrating the gradient. 
However, it is rather difficult to compute the corresponding mean or to analyze it. For this reason, we introduce an approximation, $\mathcal{I}_{h,\lambda}(\theta)$, to the criterion $\mathcal{\breve{I}}_{h,\lambda}$
and $\breve{\theta}_{h,\lambda}$ valid for small $h$ defined as
{\footnotesize\begin{eqnarray}
&&\hspace*{-1cm}\frac{1}{N}\sum_{i=1}^{N}\int_{0}^{1}\int_{0}^{1}K_{h}(t-s) \| \log\left(\exp\left(-(t-s)A_{\theta_{i}}\left(\frac{s+t}{2}\right)\right)\exp\left((t-s)A_{\theta}\left(\frac{s+t}{2}\right)\right)\right)\| _{F}^{2}dsdt \nonumber \\
&& \quad +\lambda\int_{0}^{1}\| \theta^{''}(t)\| ^{2}dt\label{eq:SelfPrediction2}
\end{eqnarray}}
\begin{prop}\label{ApproxCrt}
Let $Q_{1},\dots,Q_{N}$ be Frenet paths with parameters $\theta_{i},\:i=1,\dots,N$
in $\mathcal{H}$, satisfying $\| \theta_{i}\| _{2}^{2}\leq\frac{\pi}{2}$. 
Then, there exists $B>0$, such that for all $\| \theta\| _{2}\leq B$,
\[
\breve{\mathcal{I}}_{h,\lambda}(\theta)-\mathcal{I}_{h,\lambda}(\theta)=O(h^{3}).
\]
\end{prop}
This shows that, at first approximation, our approach is tractable and can be easily understood in terms of the geometry the curves. 

\section{Statistical Estimation}\label{sec:estimation}

In practice, the Frenet paths are rarely directly observed and
they usually have to be estimated from the available data. Typically,
we observe $N$ multidimensional trajectories $x_{1},\dots,x_{N}$,
and the Frenet frames can be computed by Gram-Schmidt orthonormalization
of the time derivatives $x_{i}^{(k)},\,k=1,2,\dots,p$. Although this
approach is easy to implement in any dimension, some care needs to be
taken as the derivatives are estimated from discretely observed trajectories
$x_{i}(t_{ij}), i=1,\ldots,N, j=1,\ldots,n_i$, which might have been corrupted by noise.
We propose a simple statistical model for accounting for random errors in the Frenet frames due to computation and measurement errors. Within this statistical
setting, we propose a statistical methodology for obtaining
estimates of the individual shapes $\theta_{i}$ and of the mean shape
$\breve{\theta}_{h,\lambda}$. For this, we develop a nonparametric statistical
model in the Lie group $SO(3)$ that will help us to analyze the variations
of the geometry of the trajectories $X_{i}$. 

\subsection{A simple model for noisy and discrete observations\label{sub:StatisticalModelNoisyObservations}}

We consider a population of Frenet paths $Q_{1},\dots,Q_{N}$ in $\mathcal{F}$,
with curvature parameter $\theta_{i}$ in $\mathcal{H}$ and initial
condition $Q_{i}^{0}$. Contrary to the situation described in section
\ref{sec:FDA-ShapeAnalysis}, we do not consider that the initial
condition is known or that it is the same for all the paths. Our objective
is to retrieve and identify from noisy observations the two sources
of variations in that population: the unknown initial conditions $\boldsymbol{Q^{0}}=\left\{ Q_{i}^{0}\in SO(3),\,i=1,\dots,N\right\} $,
the functional parameters $\boldsymbol{\theta}=\left\{ \theta_{i}\in\mathcal{H}\right\} $,
and also the mean shape $\breve{\theta}_{h,\lambda}$ as defined in
proposition \ref{ExistenceUniquenessMean}. 

We assume that we have discrete and noisy observations $U_{ij}$ of
the Frenet paths $Q_{i}$ at $n_{ij}$ positions $s_{ij},\:j=1,\dots,n_{i}$
such that $0=s_{i1}<s_{i2}\dots<s_{in_{i}}=1$, i.e for $i=1\dots,N$,
we have 
\begin{equation}
U_{ij}=Q_{i}(s_{ij})\exp\left(W_{ij}\right)\label{eq:StatisticalModelNoise}
\end{equation}
where $W_{ij}$ are random matrices in $\mathfrak{so}(p)$ (i.e $\exp(W_{ij})$
is a random rotation). For simplicity, we assume that for all $i,j$,
the matrices $W_{ij}$ are independent and identically distributed,
such that $E\left[W_{ij}\right]=0$ and $E\left[\exp\left(W_{ij}\right)\right]=I_{p}$,
such that $E\left[\| W_{ij}\| _{F}^{2}\right]=\sigma_{W}^{2}<\infty$.
We assume that the density of the noise is unknown but fixed. The standard distributions on $SO(3)$ satisfy these assumptions, such
as the Fisher-Langevin matrix distribution $\mathcal{F}(M,\alpha)$,
which possesses the density $f(U;D,\alpha)=\left\{ _{0}F_{1}(\frac{p}{2};\frac{\alpha^{2}}{4})\right\} ^{-1}\exp\left(\alpha\mbox{Tr}\left(D^{\top}U\right)\right)$
with respect to the uniform distribution on $SO(3)$ \citep{Mardia1999}.
The parameter $D$ in $SO(3)$ is the mean direction and $\alpha$
is the concentration parameter. Other distributions, with densities
$f(\cdot)$ defined with spectral functions, satisfy these properties,
see e.g. \citet{Boumal2014}. 
If we consider the Fisher-Langevin distribution, a common approach
to parameter inference is to maximize the log-likelihood of the sample
$\boldsymbol{U}=\left\{ U_{ij}\right\} _{i,j=1,\dots,n,N}$ defined
by 
\begin{eqnarray*}
\mathcal{L}\left(\boldsymbol{U};\boldsymbol{\theta},\boldsymbol{Q^{0}}\right) & = & -\log\left(_{0}F_{1}(\frac{p}{2};\frac{\alpha^{2}}{4})\right)+\mbox{Tr}\left(\frac{1}{n_{i}N}\sum_{i,j=1}^{N,n_{i}}Q_{i}(s_{ij})^{\top}U_{ij}\right).
\end{eqnarray*}
The Maximum Likelihood Estimator (MLE) of the paths is defined as
the solution of the constrained optimization problem
\begin{equation}
\begin{array}{l}
\min_{\boldsymbol{\theta},\boldsymbol{Q^{0}}}\sum_{i=1}^{N}\sum_{j=1}^{n_{i}}\| U_{ij}-Q_{i}(s_{ij})\| _{F}^{2}\\
\mbox{s.t.}Q_{i}^{'}(s)=Q_{i}(s)A_{\theta_{i}}(s)\,;\,Q_{i}(0)=Q_{i}^{0}
\end{array}\label{eq:MLE_kappa_tau-1}
\end{equation}
Our statistical problem can be seen as the parameter estimation of
an Ordinary Differential Equation from noisy data, with time-varying
parameters (i.e a semiparametric estimation problem). It is well-known that parameter estimation in ODEs can be ill-posed, and that the classical approach (e.g., MLE, least-squares etc.) can lead to unstable and unreliable results. Satisfactory estimators based on the relaxation of the ODE constraint have been developed by mixing nonparametric estimation and numerical analysis techniques, see \citet{Ramsay2007,WuDing2014}.
In particular, the observations $U_{ij}$ are smoothed by integrating
a constraint on the derivatives: the parameter estimation is achieved
by simultaneously estimating the solution path and the parameters
through an iterative scheme. 

We follow a similar line, and we build on the analysis developed in
section \ref{sub:Mean-shape-ODE} by replacing the Frenet paths $Q_{i}$
by a nonparametric proxy. This approach permits us to take into account
two specific constraints arising in our estimation problem: the $SO(3)$
manifold constraint, and the estimation of a functional parameter
$\theta$. Moreover, our approach is appropriate for estimating $\theta_{i}$
and $Q_{i}$ in the case of a single curve, as for estimating the
mean shape $\breve{\theta}_{h,\lambda}$ of a population. 

\subsection{Statistical Criterion}

The parameter estimation problem can be cast into a prediction problem
by building on the criteria (\ref{eq:SelfPrediction1}) and the definition
of the mean vector field. A first step in the parameter estimation
problem is to recover the Frenet path $Q_{i}$ from the noisy observations
$U_{ij},\,j=1,\dots,n_{i}$. Instead of using a Frenet path, we look
for a path $t\mapsto M_{i}(t)$ that is close to the data and that approximately satisfies  the self-prediction property (\ref{eq:selfprediction}):
for all $j$, and for all $t$ in $\left[0,1\right]$, we should have
approximately $\phi_{\theta_{i}}\left(t-s_{ij},s_{ij},U_{ij}\right)\approx Q_{i}(t)$.
Hence, if we combine all the predictions by averaging $\frac{1}{n_{i}}\sum_{j=1}^{n_{i}}\phi_{\theta_{i}}\left(t-s_{ij},s_{ij},U_{ij}\right)$,
we should reduce the prediction error due to the propagation of the
noise by the flow $\phi_{\theta_{i}}$ (we suppose here the the flow
$\phi_{\theta_{i}}$ is exactly known). Obviously, the arithmetic
average is not adapted, as it is implicitly based on Euclidean assumptions.
Instead, we use a Karcher mean for defining the best prediction at
time $t$: 
\begin{equation}
M_{i}(t)=\arg\min_{M\in SO(3)}\frac{1}{n_{i}}\sum_{j=1}^{n_{i}}K_{h}(t-s_{ij})\| \log\left(M^{\top}\phi_{\theta}\left(t-s_{ij},s_{ij},U_{ij}\right)\right)\| _{F}^{2}\label{eq:ODESmoothing}
\end{equation}
where the kernel $K_{h}$ accounts for the increasing uncertainty
for distant points. This approach is close to the local polynomial
smoothing, and it is in the same vein as the adaptation of smoothing
to manifolds and other exotic spaces (where usual least squares estimates
are replaced by Fr\'echet or Karcher means, see \citet{Bhattacharya2014,Jakubiak2006,Samir2012}).
The computation of the Karcher mean in $SO(3)$ such as defined in
(\ref{eq:ODESmoothing}) is a classical problem. In the case of the
Frobenius distance, the Karcher mean is simply the polar decomposition
of the arithmetic sample mean of the rotation matrices. In the case
of the geodesic distance, we do not benefit from a closed-form expression;
moreover, the first question to address is the existence and the uniqueness
of the Karcher mean in $SO(3)$, as the Karcher mean is in general not unique. 
Nevertheless, in the case of $SO(3)$, we can guarantee the existence and uniqueness of the
Karcher mean, and we can also provide a simple and efficient gradient
algorithm for computing it. 
\begin{prop}\label{KarcherMeanSOp}
Let $V_{1},\dots,V_{N}$ observations in $SO(3)$,
with positive weights $\omega_{j}$, such that they are all in a ball
of radius $\frac{1}{2}\min\left(\mbox{inj}(SO(3)),\frac{\pi}{\sqrt{k}}\right)$,
where $\mbox{inj}(SO(3))$ is the injectivity radius and $k$ is the
sectional curvature of $SO(3)$ then there exists a unique Karcher
mean defined as 
\begin{equation}
\bar{V}=\arg\min_{M\in SO(3)}\frac{1}{2}\sum_{j=1}^{n}\omega_{j}\| \log\left(M^{\top}V_{j}\right)\| _{F}^{2}\label{eq:MeanSOpGeodesic}
\end{equation}
The gradient of $F(M)=\frac{1}{2}\sum_{j=1}^{n}\omega_{j}\| \log\left(V_{j}^{\top}M\right)\| _{F}^{2}$
is 
\[
\mbox{grad}F(M)=\sum_{j=1}^{n}\omega_{j}\log\left(M^{\top}V_{j}\right)
\]
and the sequence defined for $k\geq1$, $V_{k+1}=V_{k}\exp\left(-\mbox{grad}F(V_{k})\right)$
converges to $\bar{V}$ for any initial guess $V_{0}$. 
\end{prop}
\begin{proof}
See \citet{Le2004,Rentmeesters2011}.
\end{proof}
Thanks to proposition \ref{KarcherMeanSOp}, we can define the simple
kernel smoother for $t$ in $\left[0,1\right]$ 
\begin{equation}
\widetilde{M}_{i}(t)=\arg\min_{M\in SO(p)}\sum_{j=1}^{n_{i}}K_{h}(t-s_{ij})\| \log\left(M^{\top}U_{ij}\right)\| _{F}^{2}\label{eq:SmootherSOp}
\end{equation}
that can estimate each path $Q_{i}$. More generally we can use
a model-based smoother $\hat{M}_{i}(t;\theta)$ by solving (\ref{eq:ODESmoothing})
for every $t$ and a given parameter candidate $\theta$. 

The smoothing (\ref{eq:ODESmoothing}) however relies on a candidate
parameter $\theta$ that might be inappropriate. For this reason,
we propose to simultaneously estimate the approximate Frenet paths
$\boldsymbol{M}=\left\{ M_{1},\dots,M_{N}\right\} $ and the mean
parameter $\theta$ by minimizing 
{\footnotesize\begin{equation}
\hspace*{-0.3cm}\breve{\mathcal{J}}_{h,\lambda}\left(\theta,\boldsymbol{M};\boldsymbol{U}\right)=\sum_{i=1}^{N}\frac{1}{n_{i}}\sum_{j=1}^{n_{i}}\int_{0}^{1}\petit K_{h}(t-s_{ij})\| \log\left(M_{i}(t)^{\top}\phi_{\theta}\left(t-s_{ij},s_{ij},U_{ij}\right)\right)\| _{F}^{2}dt+\lambda\petit \int_{0}^{1}\petit \| \theta^{''}(t)\| ^{2}\petit dt\label{eq:EmpiricalGlobalCrit}
\end{equation}}
The paths $M_{i}$ are smooth functions from $\left[0,1\right]$ to
$SO(p)$, and that function space is denoted by $SSO(p)$. We have seen
in the previous section that we can avoid the use of the exact flow
defined with $\Omega(t-s,s;\theta)$, by replacing it with a first order
approximation to the Magnus expansion: 
{\footnotesize\begin{eqnarray}
&&\hspace*{-1cm}\mathcal{J}_{h,\lambda}\left(\theta,\boldsymbol{M};\boldsymbol{U}\right)
=\sum_{i=1}^{N}\frac{1}{n_{i}}\sum_{j=1}^{n_{i}}\int_{0}^{1}\petit K_{h}(t-s_{ij})\| \log\left(M_{i}(t)^{\top}U_{ij}\exp\left(\left(t-s_{ij}\right)A_{\theta}\left(\frac{t+s_{ij}}{2}\right)\right)\right)\| _{F}^{2}\petit dt \nonumber\\
& &+\lambda\petit \int_{0}^{1}\petit \| \theta^{''}(t)\| ^{2}dt\label{eq:EmpiricalGlobalCrit-2}
\end{eqnarray}}
A \emph{natural} estimator can be defined as the solution to the optimization problem as 
\[\arg\min_{\theta,\boldsymbol{M}}\mathcal{J}_{h,\lambda}\left(\theta,\boldsymbol{M};\boldsymbol{U}\right) \,.
\]
In order to gain some insight into this criterion, we give an approximation of the criterion $\mathcal{J}_{h,\lambda}(\theta,\boldsymbol{M};\mathbf{Q})$ for small $h$. 
\begin{prop}[Approximation of the criteria]%
\label{ApproximationDeterministicCrit} 
Let $\mathbf{Q}=Q_{1},\ldots,Q_{N}$ be Frenet paths associated with
$\theta_{1},\dots,\theta_{N}$. For any functions $M_{1},\dots,M_{N}$
in $SSO(3)$, $\theta$ in $\mathcal{H}$, and small $h$, we have
the approximation $\mathcal{J}_{h,\lambda}(\theta,\boldsymbol{M};\mathbf{Q})=\boldsymbol{J}_{h,\lambda}(\theta,\boldsymbol{M};\mathbf{Q})+O(h^{3}+\frac{1}{n^{2}})$
where $\boldsymbol{J}_{h,\lambda}(\theta,\boldsymbol{M};\mathbf{Q})$ is given by
{\footnotesize\begin{eqnarray*}
&& \sum_{i=1}^{N}\iint_{\left[0,1\right]^{2}}K_{h}(t-s)\| r_{i}(t)\| _{F}^{2}dtds+\mu_2(K)h^{2}\int_{0}^{1}\| \left(A_{\theta}-A_{\theta_{i}}\right)\left(t\right)+\left[r_{i}(t),A_{\theta}\left(t\right)\right]\| _{F}^{2}dt\\
& & +h^{2}\mu_2({K})\frac{\left\langle r_{i}(1),\left(A_{\theta}-A_{\theta_{i}}\right)(1)+\left[r_{i}(1),A_{\theta}\left(1\right)\right]\right\rangle +\left\langle r_{i}(0),\left(A_{\theta}-A_{\theta_{i}}\right)(0)+\left[r_{i}(0),A_{\theta}\left(0\right)\right]\right\rangle }{2} \\
& &+\lambda\int_{0}^{1}\| \theta^{''}\| ^{2}dt
\end{eqnarray*}}
and $r_i(t)\triangleq \log\left(M_i(t)^\top Q_i(t)\right)$ is the discrepancy between $Q_i$ and $M_i$. 
\end{prop}
If the true paths $Q_{i}$ are replaced by noisy observations $U_{ij}$,
then the next proposition shows that the criterion $\mathcal{J}_{h,\lambda}(\theta,\boldsymbol{M};\mathbf{U})$
converges in probability as $n$ tends to infinity towards a perturbed
function $\boldsymbol{J}_{h,\lambda}(\theta,\boldsymbol{M};\mathbf{Q})$
that still can discriminate the true path and the true parameter (when
$h$ is small).
Convergence of the empirical criterion $\mathcal{J}_{h,\lambda}(\theta,\boldsymbol{M};\mathbf{U})$ is stated below.
\begin{prop}[Convergence of the empirical criterion $\mathcal{J}_{h,\lambda}(\theta,\boldsymbol{M};\mathbf{U})$] \label{ApproximationStoCrit}
Let $\mathbf{U}$ be random observations coming from \ref{eq:StatisticalModelNoise},
and $\theta_{1},\dots,\theta_{N}$ candidate parameters in $\mathcal{H}$.
We assume that $E\| W_{ij}\| _{F}^{2}=\rho_{2,W}<\infty$
and $E\| W_{ij}\| _{F}^{4}=\rho_{4,W}<\infty$. For
any functions $M_{1},\dots,M_{N}$ in $SSO(3)$, and small $h$, we
have 
\begin{eqnarray*}
\mathcal{J}_{h,\lambda}(\theta,\boldsymbol{M};\mathbf{U}) & = & \boldsymbol{J}_{h,\lambda}(\theta,\boldsymbol{M};\mathbf{Q})+\rho_{2,W}+\mu_2(K)h^{2}\int_{0}^{1}E\left(Z_2(s)\right)ds+O_{p}(h^{3})
\end{eqnarray*}
where $Z_2(s)=\| \left[W,A_{\theta}(s)\right]\| _{F}^{2}$.
\end{prop}

The propositions show that the stochastic criterion is close to the
function $\boldsymbol{J}_{h,\lambda}(\theta,\boldsymbol{M};\mathbf{Q})$
that corresponds to a weighted distance between $M_{i}(t)$ and $Q_{i}(t)$,
plus a distance between $\theta$ and $\bar{\theta}$, with an additional
penalty term related to $\| \theta\| _{L^{2}}^{2}$. 
Nevertheless, the simultaneous optimization in $(\theta,\boldsymbol{M})$
is rather difficult, and we propose an alternating optimization scheme
described below. 

\subsection{Estimation algorithm}

For a given parameter $\theta$, the path $\boldsymbol{M}$
is estimated by solving 
\begin{equation}
\widehat{\boldsymbol{M}}_{h,\lambda}(\cdot;\theta)=\arg\min_{\boldsymbol{M}\in SSO(p)}\mathcal{J}_{h,\lambda}\left(\theta,\boldsymbol{M};\boldsymbol{U}\right)\label{eq:Profiling_M}
\end{equation}
The estimation of $\mbox{\ensuremath{\theta}}$ from a given path
$\boldsymbol{M}$ is done by solving
\begin{equation}
\widehat{\theta}_{h,\lambda}\left(\cdot;\boldsymbol{M}\right)=\arg\min_{\theta}\mathcal{J}_{h,\lambda}\left(\theta,\boldsymbol{M};\boldsymbol{U}\right)\label{eq:Profiling_theta}
\end{equation}
We need to solve the nonparametric estimation problem (\ref{eq:Profiling_theta}),
but in practice, we discretize the integral on a grid $0=t_{i1}<t_{i2}<\dots<t_{iQ_{i}}=1$
and we solve 
{\footnotesize\[
\min_{\theta\in\mathcal{H}}\sum_{i=1}^{N}\frac{1}{n_{i}}\sum_{j,q=1}^{n_{i},Q_{i}}K_{h}(t_{iq}-s_{ij})\| \log\left(M_{i}(t_{iq}){}^{\top}U_{ij}\exp\left(\left(t_{iq}-s_{ij}\right)A_{\theta}\left(\frac{t_{iq}+s_{ij}}{2}\right)\right)\right)\| _{F}^{2}+\lambda\| \theta^{''}\| _{L^{2}}^{2}
\]}
The presence of the exponential makes the optimization difficult,
and we use an additional approximation that provides a simple algorithm
and simplifies the analysis of our estimator. We define the skew-symmetric
matrix $\tilde{R}_{ijq}=-\frac{1}{t_{iq}-s_{ij}}\log\left(M_{i}(t_{iq})^{\top}U_{ij}\right)$.
The BCH formula in (\ref{eq:BCHformula}) furnishes a first order approximation, and we introduce
a new approximate criterion, with $u_{ijq}=t_{iq}-s_{ij}$, $v_{ijq}=\frac{t_{iq}+s_{ij}}{2}$
\[
\tilde{\mathcal{J}}_{h,\lambda}\left(\theta,\boldsymbol{M};\boldsymbol{U}\right)=\sum_{i,j,q=1}^{N,n_{i},Q_{i}}\frac{1}{n_{i}Q_{i}}K_{h}(u_{ijq})u_{ijq}^{2}\| A_{\theta}(v_{ijq})-\tilde{R}_{ijq}\| _{F}^{2}+\lambda\int_{0}^{1}\| \theta^{''}\| ^{2}dt
\]
In the particular case $p=3$, if we define
\[
\tilde{R}_{ijq}=\left[\begin{array}{ccc}
0 & -r_{ijq}^{1} & -r_{ijq}^{3}\\
r_{ijq}^{1} & 0 & -r_{ijq}^{2}\\
r_{ijq}^{3} & r_{ijq}^{2} & 0
\end{array}\right]
\]
the Frobenius norm can be rearranged with weights $\omega_{ijq}=\frac{2}{n_{i}Q_{i}}K_{h}(u_{ijq})u_{ijq}^{2}$
{\footnotesize\begin{equation}
\tilde{\mathcal{J}}_{h,\lambda}\left(\theta,\boldsymbol{M};\boldsymbol{U}\right)=\sum_{i,j,q=1}^{N,n_{i},Q_{i}}\omega_{ijq}\left(\kappa(v_{ijq})-r_{ijq}^{1}\right)^{2}+\sum_{i,j,q=1}^{N,n_{i},Q_{i}}\omega_{ijq}\left(\tau(v_{ijq})-r_{ijq}^{2}\right)^{2}+\lambda\int_{0}^{1}\| \theta^{''}\| ^{2}dt\label{eq:I1_g}
\end{equation}}
The optimization problem 
\[
\tilde{\theta}_{h,\lambda}\left(\cdot,\boldsymbol{M}\right)=\min_{\theta\in\mathcal{H}}\tilde{\mathcal{J}}_{h,\lambda}\left(\theta,\boldsymbol{M};\boldsymbol{U}\right)
\]
gives rise to the computation of $2$ independent smoothing splines
(with splines of third order), defined at the knots $v_{ijq}$, with
the pseudo-observations $r_{ijq}^{1},r_{ijq}^{2}$. 
The only difference with respect to the classical smoothing splines
is the presence of the weights $\omega_{ijq}$. 

Finally, the successive approximations invites us to propose the following
estimation algorithm for the Frenet Paths and the mean shape $\theta$,
which can be seen as a first order minimization of the criterion $\mathcal{J}_{h,\lambda}\left(\theta,\boldsymbol{M};\boldsymbol{U}\right)$. 

\paragraph{Estimation Algorithm\label{par:Estimation-Algorithm}}

Let us start with a given initial value $\hat{\boldsymbol{M}}^{(0)}$
(close to the real paths $\boldsymbol{M}$) and solve 
\[
\tilde{\theta}^{(0)}\left(\cdot;\hat{\boldsymbol{M}}^{(0)}\right)=\arg\min_{\theta}\tilde{\mathcal{J}}_{h,\lambda}\left(\theta,\hat{\boldsymbol{M}}^{(0)};\boldsymbol{U}\right).
\]
For a fixed $h$ and $\lambda$ and $\ell\geq1$, repeat steps 1,
2, 3 until convergence:
\begin{enumerate}
\item For all $t$ in $\left[0,1\right]$ and $i=1,\dots,N$, compute the
Karcher mean defined by 
{\footnotesize\begin{equation}
\widehat{M_{i}}^{(\ell+1)}(t)=\arg\min_{M\in SO(p)}\sum_{j=1}^{n_{i}}K_{h}(t-s_{ij})\| \log\left(M^{\top}U_{ij}\exp\left(\left(t-s_{ij}\right)A_{\tilde{\theta}^{(\ell)}}\left(\frac{t+s_{ij}}{2}\right)\right)\right)\| _{F}^{2}\label{eq:Karcher-Step1}
\end{equation}}
\item Compute the current ``empirical $\Omega$'': $\tilde{R}_{ijq}^{(\ell+1)}=-\frac{1}{t_{iq}-s_{ij}}\log\left(\widehat{M_{i}}^{(\ell+1)}(t_{iq})^{\top}U_{ij}\right)$. 
\item For all $k=1,\dots,p-1$, compute the smoothing splines 
\[
\tilde{\theta}_{k}^{(\ell+1)}=\arg\min_{\theta\in\mathcal{H}}\sum_{i,j,q=1}^{N,n_{i},Q_{i}}\omega_{ijq}\left(\theta(v_{ijq})-r_{ijq}^{k,(\ell+1)}\right)^{2}+\lambda\int_{0}^{1}\| \theta^{''}\| ^{2}dt
\]
and $\widehat{\theta}^{(\ell+1)}=\left(\tilde{\theta}_{1}^{(\ell+1)},\dots,\tilde{\theta}_{p-1}^{(\ell+1)}\right)^{\top}$.
\end{enumerate}
The outputs of the algorithm are our estimators of the curvatures
and Frenet Paths and are denoted by $\left(\hat{\theta}_{h,\lambda},\hat{\boldsymbol{M}}_{h,\lambda}\right)$. 

\begin{rem}
Our prediction error depends on $h,\lambda=(\lambda_1,\lambda_2)$. 
If $h$ is too big, we integrate along the whole interval and the errors
accumulate, and it is better to restrict to smaller interval. We consider
the prediction of a small percentage (10\%, $h\approx0.1$) of the
individuals, when the total length of a curve is 1. 
In our numerical studies we have performed 10-fold cross validation by minimizing
\[
\sum_{k=1}^{K}\sum_{(i,j)\in T_{k}}\left\|\log\left(U_{ij}^\top \hat M_i^{-(k)}(s_{ij}; h,\lambda)\right)\right\|_F^2
\]
where $T_k$ is the $k$th index set based on $K=10$ random partition of the observations $\left\{U_{ij}\right\}$,  $i=1,\ldots,N, j=1, \ldots,n_i$ and $\hat{M}_{i}^{-(k)}(s_{ij}; h,\lambda)$ are the Frenet paths estimated without the $k$th partition dataset, using hyperparameters $h, \lambda$. 
\end{rem}

\section{Numerical studies} \label{sec:numeric}

We conduct limited simulation studies to assess performance of the proposed methods, focusing on understanding the effect of pre-processing in the case of indirectly observed Frenet paths and the effect of tuning parameters with respect to the sample size. 

Our reference shape is defined by
\[
\begin{cases}
\bar\kappa(s)= & \exp(\zeta\sin(s))\\
\bar\tau(s)= & \eta s-0.5
\end{cases},\:\zeta=1,\,\eta=0.2 \,.
\]
and we set $s \in [0, 5]$.
We consider the estimation problem for both a single curve and multiple curves with directly observed Frenet paths and indirectly observed Frenet paths from Euclidean curves under this reference model. Finally, we add a case of the Euclidean curves with unknown parameter model. All simulation models are repeated for 100 times. We have not run an exhaustive search for the best hyperparameters but run standard 10 fold cross validation in a small selected parameter set for each case to reduce computational cost by

\subsection{Observations are a single noisy Frenet Path \label{subsec:SingleFrenetPath}} 

Given $\theta = (\bar\kappa,\bar\tau)$, the observation model is defined as 
\[
U_{j}=Q(s_{j})M_{j}\,, \quad j=1,\ldots, n\,,
\] 
where $Q$ is the solution to $\dot{Q}=A_{\theta}Q$ and $Q(0)=Q^{0}$ with $Q^{0}\sim\mathcal{F}(I_{3},\alpha)$ and random rotations $M_{j} \sim\mathcal{F}(I_{3},\alpha)$, Fisher-Langevin distribution with mean identity and concentration $\alpha$. We assume that $Q(0)$ is also unknown. 

We consider different sample sizes $n=100,200$ and noise levels $\alpha=5,10$. We select the hyperparameters in $h\in\left\{ 0.3,0.5,0.7\right\}$ and $\lambda_{1},\lambda_{2}\in\left\{ 10^{k}\vert\,k=1,0,-1,-2\right\}$ by cross validation. 
For evaluating the quality of the Lie smoother, we compare the distance
between the noisy data and the true Frenet path ($\Delta_{O}$) and the distance between the true and the smoothed Frenet path ($\Delta_{FS}$) :
\[
\Delta_O=\frac{1}{n}\sum_{i=1}^{n}d\left(U_{j},Q(s_j)\right)\,,\quad \Delta_{FS}=\frac{1}{n}\sum_{i=1}^{n}d\left(\hat{Q}^{FS}(s_j),Q(s_j)\right) \,.
\]
We also compute the $L^{2}$ distance between the curvature and torsion estimates
$\hat{\kappa}^{FS},\hat{\tau}^{FS}$ and the true parameters.
The results are summarized in Table \ref{tab:S1.1}.

\begin{table}
\caption{\label{tab:S1.1} Estimation error from a single noisy Frenet path.}
\centering
\begin{tabular}{|c|c|c|c|c|c|}
\hline 
$n$ & $\alpha$ & $\Delta_{O}$ & $\Delta_{FS}$ & $\|\hat{\kappa}^{FS}-\bar\kappa\|_{L^{2}}^{2}$ & $\|\hat{\tau}^{FS}-\bar\tau\|_{L^{2}}^{2}$ \\ 
\hline 
\hline 
$100$ & $5$ & 0.74 (0.03) & 0.18 (0.03) & 0.32 (0.26) & 0.15 (0.22)\\ 
\hline 
$100$ & $10$ & 0.51 (0.02) & 0.13 (0.02) & 0.18 (0.23) & 0.06 (0.09)\\ 
\hline 
$200$ & $5$ & 0.74 (0.03) & 0.13 (0.02) & 0.19 (0.16) & 0.09 (0.14)\\ 
\hline 
$200$ & $10$ & 0.51 (0.02) & 0.09 (0.02) & 0.12 (0.12) & 0.04 (0.08)\\ 
\hline 
\end{tabular}
\end{table}

\subsection{Observations are a single noisy Euclidean Curve\label{subsec:SingleEuclideanCurve}}

We observe noisy observations $y_j \in \mathbb{R}^3$ where 
\begin{equation}
y_{j}=X(s_{j})+\sigma\epsilon_{j},\:j=1,\dots,n\,. \label{eq:DGP_X_SIngle}
\end{equation}
Here $s\mapsto X(s)$ has a Frenet path $Q(s)$ solution of the ODE
$Q^{\prime}(s)=A_{\theta}(s)Q(s)$.
We consider sample sizes $n=100,\,200$ and $\sigma=0.05$ or $\sigma=0.02$. 
An added difficulty with this setting is related to defining a preliminary estimate of the Frenet path.
As a preprocessing step, we nonparametrically estimate the higher-order derivatives of $X$, $X^{(k)}, k=1,2,3$ from the noisy observations $y_{1},\dots,y_{n}$. These derivatives can be very noisy and are used for computing raw estimates $\hat{Q}(s_{j}),\,j=1,\dots,n$. We consider
two methods for deriving these estimates: 
\begin{description}
\item [{$\hat{Q}^{GS}(s)$}] obtained by Gram-Schmidt orthonormalization
of the frame $\left[X^{(1)}\vert X^{(2)}\vert X^{(3)}\right]$. The
derivatives are estimated by a standard local polynomial of order
4. With the same derivative estimates, we can compute the estimators
of the curvature $\hat{\kappa}^{Ext}$ and torsion $\hat{\tau}^{Ext}$
using the extrinsic formulas.
\item [{$\hat{Q}^{LP}(s)$}] obtained by constrained
nonparametric smoothing of $X$. Instead of the standard local polynomial,
we use a local expansion that uses the orthogonal vectors $T,N,B$:
{\small\[
X(s+h)=X(s)+\left(h-\frac{h^{3}\kappa^{2}(s)}{6}\right)T(s)+\left(\frac{h^{2}\kappa(0)}{2}+\frac{h^{3}\kappa^{\prime}(s)}{6}\right)N(s)+\frac{h^{3}\kappa(s)\tau(s)}{6}B(s)+o(s^{3})
\]}
\end{description}
The latter is used to construct our estimator $\hat{Q}^{FS}$ as well as $\hat{\kappa}^{FS}$ and $\hat{\tau}^{FS}$. For comparison, we also compute the corresponding extrinsic estimators $\hat{\kappa}^{Ext}$ and $\hat{\tau}^{Ext}$. The hyperparameters are selected from $h\in\left\{ 0.3,0.5\right\} $, $\lambda_{1}\in\left\{ 10^{k}\vert k=0,-1,-2,-3\right\} $
and $\lambda_{2}\in\left\{ 10^{k}\vert k=2,1,0,-1\right\}$ by cross validation. The results are summarized in Table~\ref{tab:S1.2}. The comparison of $\hat{Q}^{FS}$ with $\hat{Q}^{GS}$ shows that our procedure with Frenet-Serret equation really smooths the $Q_{i}$'s and gives a stable estimate of curvature and torsion.

\begin{table}
\caption{\label{tab:S1.2} Estimation error from a single noisy Euclidean curve. Outliers are removed in extrinsic estimates.}
\centering
\begin{tabular}{|c|c|c|c|c|c|c|c|} \hline
$n$ &$\sigma$ & $\Delta_{GS}$ & $\Delta_{FS}$ & $\|\hat{\kappa}^{Ext}-\bar\kappa\| _{L^{2}}^{2}$ & $\|\hat{\kappa}^{FS}-\bar\kappa\|_{L^{2}}^{2}$ &$\|\hat{\tau}^{Ext}-\bar\tau\|_{L^{2}}^{2}$ & $\|\hat{\tau}^{FS}-\bar\tau\|_{L^{2}}^{2}$\\
\hline 
$100$ & $0.02$ & 0.4 (0.09) & 0.2 (0.07) & 8.3 (8.9) & 0.15 (0.08) & 145.7 (274) & 3.52 (4.3)\\
\hline 
$200$ & $0.02$ & 0.3 (0.08) & 0.17 (0.06) & 7 (10) & 0.12 (0.06)& 145 (297) & 3.48 (4.4)\\
\hline 
$100$ & $0.05$ & 0.7 (0.11) & 0.37 (0.14) & 20.6 (12.8) & 0.43 (0.3)& 4605 (2082) & 7 (7.7)\\
\hline 
$200$ & $0.05$ & 0.54 (0.1) & 0.29 (0.1) & 13.2 (11.3) & 0.36 (0.23)& 267 (7.2) & 7.2 (7.5)\\
\hline 
\end{tabular}
\end{table}

\subsection{Observations are a population of noisy Frenet Paths\label{subsec:PopulationFrenetPath}}

We simulate a population of random Frenet paths $s\mapsto Q_{i}(s),\:i=1,\dots,N$
generated by random Frenet-Serret equations with random individual
shape parameter $\theta_{i}$ obtained as 
\begin{equation}
\begin{cases}
\kappa_{i}= & \left|\bar{\kappa}+\sigma_{\kappa}\zeta_{i}^{^{1}}\right|\\
\tau_{i}= & \bar{\tau}+\sigma_{\tau}\zeta_{i}^{2}
\end{cases}\label{eq:RandomTheta}
\end{equation}
where $\zeta_{i}^{^{1}},\,\zeta_{i}^{2}\;i=1,\dots,N$ are centered
independent Gaussian processes with (unit) Mat\'ern covariance functions\footnote{$k(s,s')=\frac{1}{\Gamma(\nu)2^{\nu-1}}\left(\frac{\sqrt{2\nu}}{\ell}\left|s-s'\right|\right)^{\nu}K_{\nu}\left(\frac{\sqrt{2\nu}}{\ell}\left|s-s'\right|\right)$}
with $\nu=\frac{5}{2}$ and characteristic length scale $\ell=1$.
This means that the random functions are twice differentiable, and
the functions $\bar{\kappa},\bar{\tau}$ are respectively the means
of the population $\left(\kappa_{i}\right)_{i=1\dots N}$ and $\left(\tau_{i}\right)_{i=1\dots N}$. 
The number of observations $n_{i}$ per curve is relatively small, $n_{i}=25$. We set $\sigma_{\kappa}=\sigma_{\tau}=0.3$.
We consider two simulation settings: 
\begin{enumerate}
\item Exact observations: $U_{ij}=Q_{i}(s_{ij})$ where $\dot{Q}_{i}=A_{\theta_{i}}Q_{i}$
and $Q_{i}(0)=Q_{i}^{0}$, with $Q_{i}^{0}\sim\mathcal{F}(I_{3},\alpha)$ with $\alpha=10$.
\item Noisy observations: $U_{ij}=Q_{i}(s_{ij})M_{ij}$ where $\dot{Q}_{i}=A_{\theta_{i}}Q_{i}$
and $Q_{i}(0)=Q_{i}^{0}$, with $Q_{i}^{0}\sim\mathcal{F}(I_{3},\alpha)$,
random rotations $M_{ij}\sim\mathcal{F}(I_{3},\alpha)$ ($\alpha=10$). 
\end{enumerate}
We compare the two estimators $\hat{\theta}^{ind}=\frac{1}{N}\sum_{i=1}^{N}\hat{\theta}_{i}$,
where $\hat{\theta}_{i}$ is the estimate of the individual curvature
$\theta_{i}$ and $\hat{\theta}^{pop}$, which is obtained by solving
the estimation criterion $\mathcal{J}(\theta,\boldsymbol{M};\boldsymbol{U})$
for the $N$ curves simultneously. Finally, we compute $\| \hat{\theta}^{ind}-\bar{\theta}\| _{L^{2}}^{2}$
and $\| \hat{\theta}^{pop}-\bar{\theta}\| _{L^{2}}^{2}$
and the the prediction errors
\[
\Delta_{ind}=\frac{1}{Nn}\sum_{i,j=1}^{N,n}d\left(\hat{Q}_{ij}^{ind},Q_{ij}\right)\:\text{and}\:\Delta_{pop}=\frac{1}{Nn}\sum_{i,j=1}^{N,n}d\left(\hat{Q}_{ij}^{pop},Q_{ij}\right).
\]
The hyperparameters are selected by cross validation from $h\in\left\{ 0.4,0.6\right\} $, $\lambda_{1}\in\left\{ 10^{-6},10^{-8},10^{-10}\right\} $ and $\lambda_{2}\in\left\{ 10^{-6},10^{-8},10^{-10}\right\}$. The results for $N=25$ is shown in Table~\ref{tab:S2.1}, showing comparable performance between the two settings.
\begin{table}
\caption{\label{tab:S2.1} Estimation error from a population of noisy Frenet paths.}
\centering
\begin{tabular}{|c|c|c|c|c|c|c|}
\hline 
 & $\| \hat{\kappa}^{ind}-\bar{\kappa}\| _{L^{2}}^{2}$ & $\| \hat{\kappa}^{pop}-\bar{\kappa}\| _{L^{2}}^{2}$ & $\| \hat{\tau}^{ind}-\bar{\tau}\| _{L^{2}}^{2}$ & $\| \hat{\tau}^{pop}-\bar{\tau}\| _{L^{2}}^{2}$ & $\Delta_{ind}$ & $\Delta_{pop}$\\
\hline 
\hline 
Exact & 0.02 (0.01) & 0.02 (0.01) & 0.017 (0.014) & 0.017 (0.014) & 19.8 (0.9) & 19.8 (0.9)\\
\hline 
Noisy & 0.14 (0.07) & 0.12 (0.04) & 0.016 (0.014)  & 0.1 (0.04) & 20 (0.74) & 19.9 (0.8)\\
\hline 
\end{tabular}
\end{table}

\subsection{Observations are a population of noisy Euclidean Curve \label{subsec:PopulationEuclideanCurves}}

We consider that the data are noisy measurements of $X_{i},\:i=1,\dots,N$ with random
individual shape parameter $\theta_{i}$ simulated as in (\ref{eq:RandomTheta}).
\[
y_{ij}=X_{\theta_i}(s_{ij})+\sigma_{e}\epsilon_{ij}, i=1,\ldots,N, j=1,\ldots,n_i \,.
\]
The $X$ curves are already arclength parametrized and the functions $\bar{\kappa},\bar{\tau}$ defines the mean shape corresponding to the population of Euclidean curves $X_{i},\:i=1,\dots,N$. In this setting,  our estimator of the mean curvature targets the mean parameter. 
We set $\sigma_{\kappa}=0.2$ and $\sigma_{\tau}=0.08$ in (\ref{eq:RandomTheta}) and assume $n_{i}=50$ and $N=25$. 
The hyperparameters are selected from $h\in\left\{ 0.3,0.5\right\} $, $\lambda_{1}\in\left\{ 1,10^{-1},10^{-2},10^{-3}\right\}$, $\lambda_{2}\in\left\{ 10,1,10^{-1},10^{-2}\right\}$ by cross validation. 
Performance of the two estimators $\hat{\theta}^{ind}=\frac{1}{N}\sum_{i=1}^{N}\hat{\theta}_{i}$ and the mean estimator $\hat{\theta}^{pop}$ are compared in Table~\ref{tab:S2.2a}. 
\begin{table}
\caption{\label{tab:S2.2a} Estimation error from a population of Euclidean curves with known curvature model.}
\centering
\begin{tabular}{|c|c|c|c|c|c|c|}
\hline 
$\sigma_{e}$ & $\| \hat{\kappa}^{ind}-\bar{\kappa}\| _{L^{2}}^{2}$ & $\| \hat{\kappa}^{pop}-\bar{\kappa}\| _{L^{2}}^{2}$ & $\| \hat{\tau}^{ind}-\bar{\tau}\| _{L^{2}}^{2}$ & $\| \hat{\tau}^{pop}-\bar{\tau}\| _{L^{2}}^{2}$ & $\Delta_{ind}$ & $\Delta_{pop}$\\
\hline 
\hline 
$0$ & 0.15 (0.01) & 0.14 (0.01) & 0.02 (0.03) & 0.03 (0.04) & 2.3 (0.22) & 2.3 (0.2)\\
\hline 
$0.05$ & 0.5 (0.05) & 0.44 (0.05) & 0.12 (.14) & 0.17 (0.2) & 2.9 (0.29) & 2.95 (0.34)\\
\hline 
\end{tabular}
\end{table}

\subsection{Euclidean curves with unknown parameter model}

Up to now, the simulation models were defined with a true mean parameter. Our final example is constructed to investigate the case when the true mean parameter is implicitly defined. 
We consider a parametric curve defined by $x_{1}(t)=\cos(at),$ $x_{2}(t)=\sin(bt),$ and $x_{3}(t)=ct$, for $t\in\left[0,5\right]$, see figure \ref{fig:ParametricCurve}. We denote by $\varphi=(a,b,c)$ the parameter. The corresponding curvature and torsions are parametric functions that depends then on $\varphi$. Individual parameters are simulated from ${\varphi}_{i}\sim\mathcal{N}\left({\varphi}_{ref},\sigma_{P}\right)$ where $\sigma_{P}>0$ is the population variability and ${\varphi}_{ref}=(1,.9,.8)$. The corresponding generalised curvatures are denoted by $\theta_{ref}$ and $\theta_i, i = 1, \ldots, N$, respectively. The measurements are then obtained from 
\[
y_{ij}=x_{i}(t_{j})+\sigma_{e}\epsilon_{ij} \,,\quad i=1,\ldots, N; \,, j=1,\ldots,n_i\,.
\]
In the simulation, we set $N=25$, with regular sampling times $n=50$ and vary the noise level by $\sigma_{P}^{2}=0.02$ or $\sigma_{P}^{2}=0.04$ and $\sigma_{e}^{2}=0.01$ or $\sigma_{e}^{2}=0.04$.
For comparison, we identify the true parameter $\theta_i$ and $\theta_{ref}$ by numerical methods.
\begin{figure}
\begin{centering}
\includegraphics[scale=0.35]{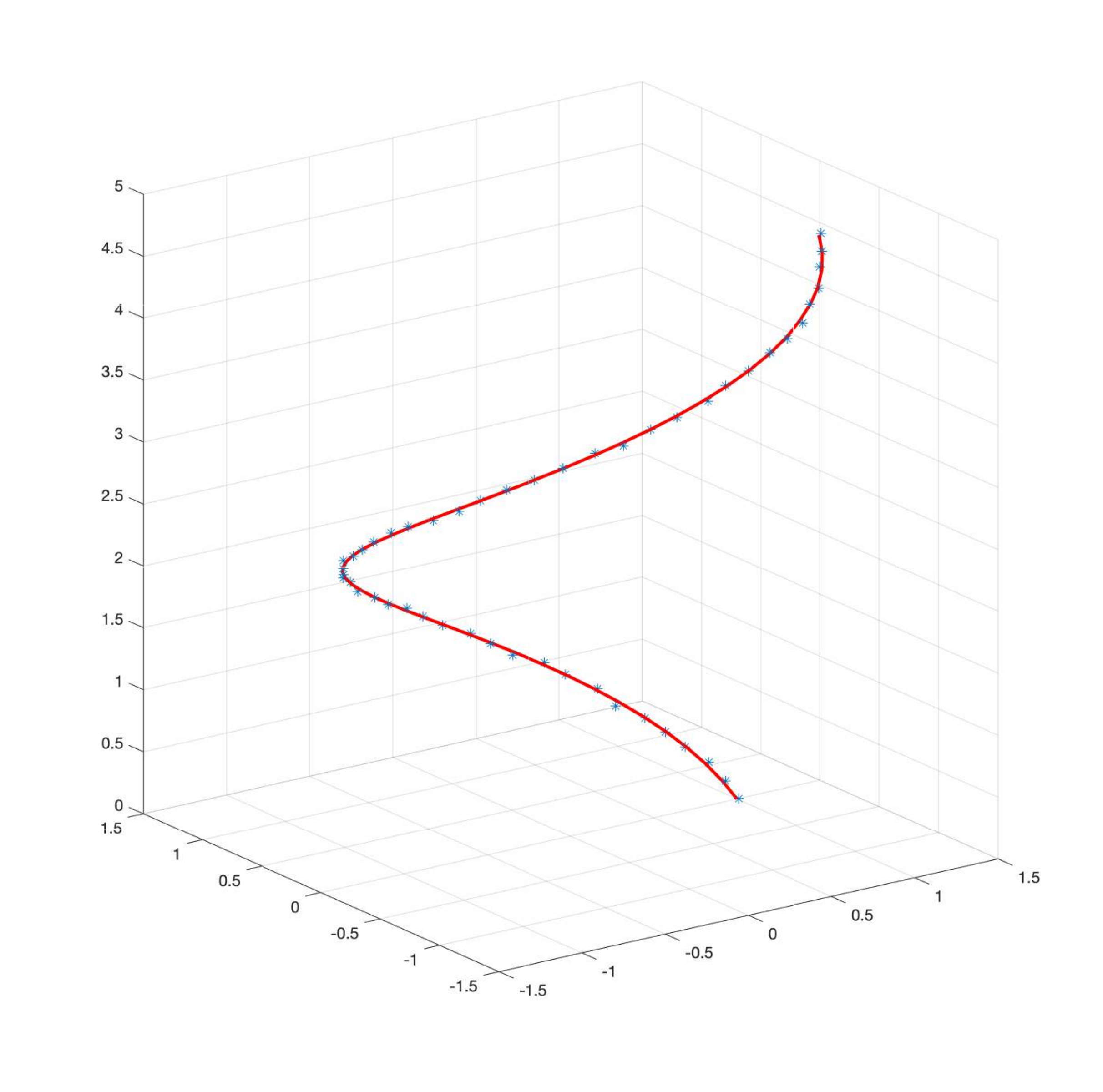}
\par\end{centering}
\caption{The parametric curve with parameters $\varphi_{ref}$, with noisy
observation points, $n=50$ and $\sigma_{e}^{2}=0.01$. \label{fig:ParametricCurve}}
\end{figure}
In general, numerical computation of $\kappa$ and $\tau$ from discretized observations on a grid is a very unstable process and it is an ill-posed problem. We compute the underlying parameters from discrete data. 

Our methodology requires arc-length parametrized curves of the same length $L$, while the curves $x_{i},i=1,\dots N$ may have different lengths $L_{i}$. For each curve $x_{i}$ , we smooth the observations $\left\{ y_{ij},j=1,\dots n\right\} $ with local polynomial and estimate arclength $s_{i}(t)$ and Frenet Path $s\mapsto Q_{i}(s)$ with the nonparametric estimates of the derivatives as previously. As the reference length $L_{ref}$, we use the smallest length, i.e $L_{ref}\triangleq\min\left\{ s_{i}(T),i=1,\dots N\right\}$ so that we have a family of curves $X_{i}$ and Frenet Path $Q_{i}$, with corresponding generalized curvature $\theta_{i}$, defined on $\left[0,L_{ref}\right]$. We define the population curvature as $\bar{\theta}\triangleq\frac{1}{N}\sum_{i=1}^{N}\theta_{i}$ on $\left[0,L_{ref}\right]$, and because of the nonlinearity, we have $\bar{\theta}\neq\theta_{ref}$ in general, compared in Figure~\ref{fig:parComp}. Nevertheless, when $\sigma_{P}$ is relatively small (i.e lower than $0.05$ in our case), the geometry of the curves varies but the main features are preserved, meaning the curvatures $\theta_{i}$ varies around $\theta_{ref}$, such that $\theta_{ref}\approx\bar{\theta}$.
\begin{figure}
\begin{centering}
\begin{tabular}{cc}
\includegraphics[scale=0.35]{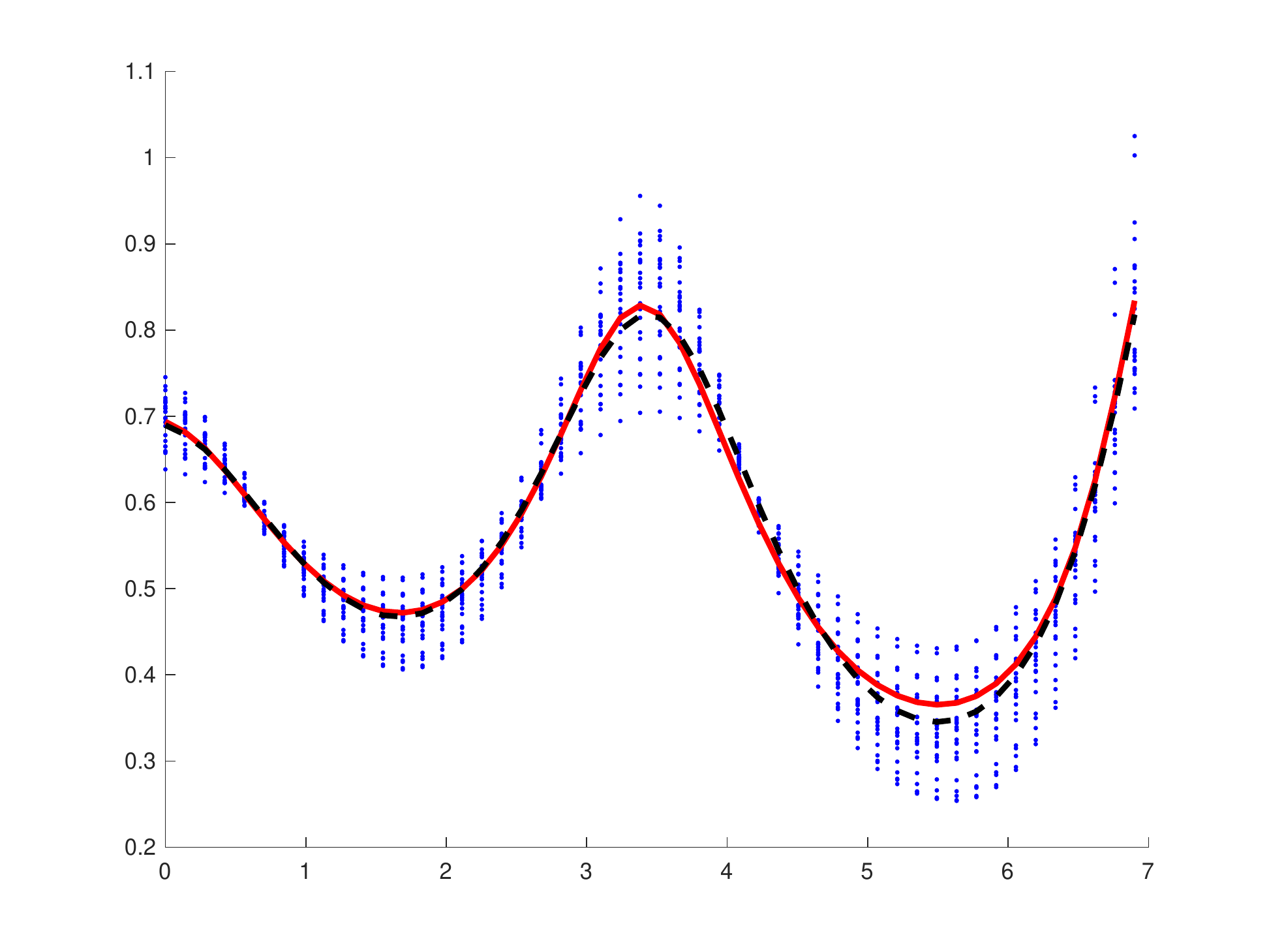} & \includegraphics[scale=0.35]{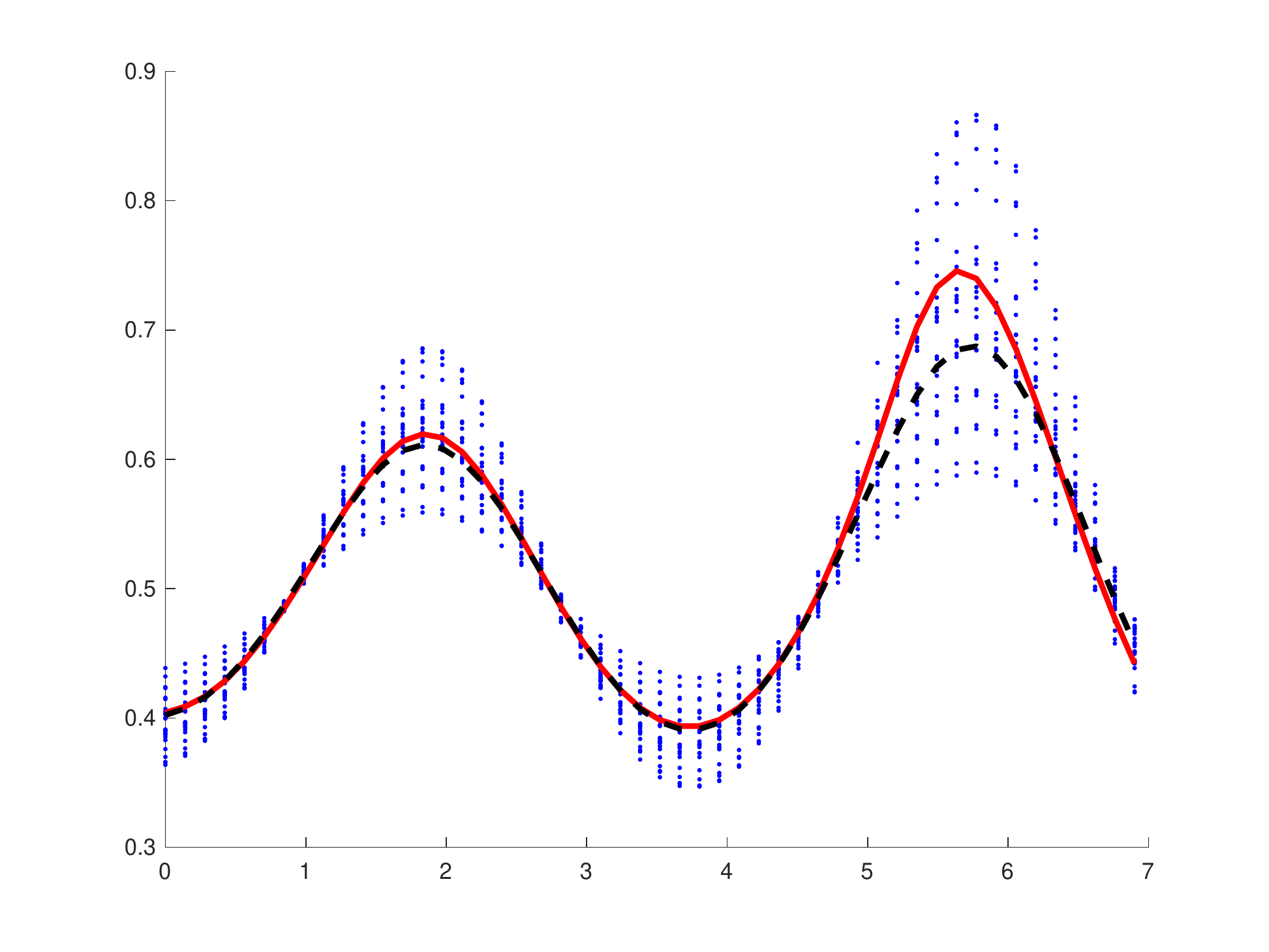}\tabularnewline
\end{tabular}
\caption{\label{fig:parComp} Difference between $\theta_{ref}$ (Dashed black line) and $\bar{\theta}$ (Solid Red Line), and the individual curvatures $\theta_{i}$ (Dotted Lines). Curvature on the Left, Torsion on the Right.}
\end{centering}
\end{figure}
We compare the mean parameter estimates $\hat{\theta}^{pop}$ and $\hat{\theta}^{ind} = (1/N) \sum_{i=1}^{N}\hat\theta_i$ with $\bar\theta$ in Table~\ref{tab:S3}. 
The hyperparameters were selected from $h=0.3$ and $\lambda_{1},\lambda_{2}\in\left\{10^{-k}, k=1, \ldots, 5\right\}$.
For comparison, we also include the empirical mean estimate $\hat{\theta}^{ind}_{Ext}$ obtained from the individual estimates $\hat{\theta}_{i}^{Ext}$ computed with the extrinsic formula. 
\begin{table}
\caption{\label{tab:S3} Estimation error from noisy Euclidean curves with unknown parameter model.}
\centering
\hspace*{-1cm}\begin{tabular}{|c|c|c|c|c|c|c|c|}
\hline 
$\sigma_{P}^{2}$ & $\sigma_{e}^{2}$ & $\| \hat{\kappa}^{ind}_{Ext}-\bar{\kappa}\| _{L^{2}}^{2}$ & $\| \hat{\kappa}^{ind}-\bar{\kappa}\| _{L^{2}}^{2}$ & $\| \hat{\kappa}^{pop}-\bar{\kappa}\| _{L^{2}}^{2}$ & $\| \hat{\tau}^{ind}_{Ext}-\bar{\tau}\| _{L^{2}}^{2}$ & $\| \hat{\tau}^{ind}-\bar{\tau}\| _{L^{2}}^{2}$ & $\| \hat{\tau}^{pop}-\bar{\tau}\| _{L^{2}}^{2}$\tabularnewline
\hline 
\hline 
\multirow{2}{*}{$0.02$} & $0.01$ & 0.06 (0.026) & 0.14 (0.028) & 0.095 (0.015) & 0.37 (1.1) & 0.03 (0.02) & 0.046 (0.02)\tabularnewline
\cline{2-8} \cline{3-8} \cline{4-8} \cline{5-8} \cline{6-8} \cline{7-8} \cline{8-8} 
 & $0.04$ & 2.3 (2.5) & 0.13 (0.02) & 0.09 (0.017) & 1.5 (1.1) & 0.08 (0.06) & 0.1 (0.08)\tabularnewline
\hline 
\multirow{2}{*}{$0.04$} & $0.01$ & 0.06 (0.03) & 0.12 (0.04) & 0.08 (0.02) & 0.32 (0.4) & 0.04 (0.02) & 0.04 (0.01)\tabularnewline
\cline{2-8} \cline{3-8} \cline{4-8} \cline{5-8} \cline{6-8} \cline{7-8} \cline{8-8} 
 & $0.04$ & 3.4 (3.2) & 0.1 (0.03) & 0.07 (0.02) & 2.6 (8.8) & 0.07 (0.05) & 0.1 (0.08)\tabularnewline
\hline 
\end{tabular}
\end{table}
The results suggest that extrinsic formula may be useful when there is a small observation noise, but it degrades fast with $\sigma_{e}^{2}$, more susceptible to the errors on the boundaries. An example of corresponding estimates is shown in Figure \ref{fig:LargeSmall-LargeHigh} comparing two noise levels when $\sigma_P^2 = 0.04$.
We can see the torsion is very hard to estimate with the extrinsic formula, but it remains stable with our proposed method, while we use exactly the same nonparametric estimates of the derivatives, but our approach removes oscillations and noise, and in addition we have a joint estimation of $\kappa$ and $\tau$ which makes the global shape more faithful.

\begin{figure}
\begin{centering}
\begin{tabular}{cc}
\includegraphics[scale=0.35]{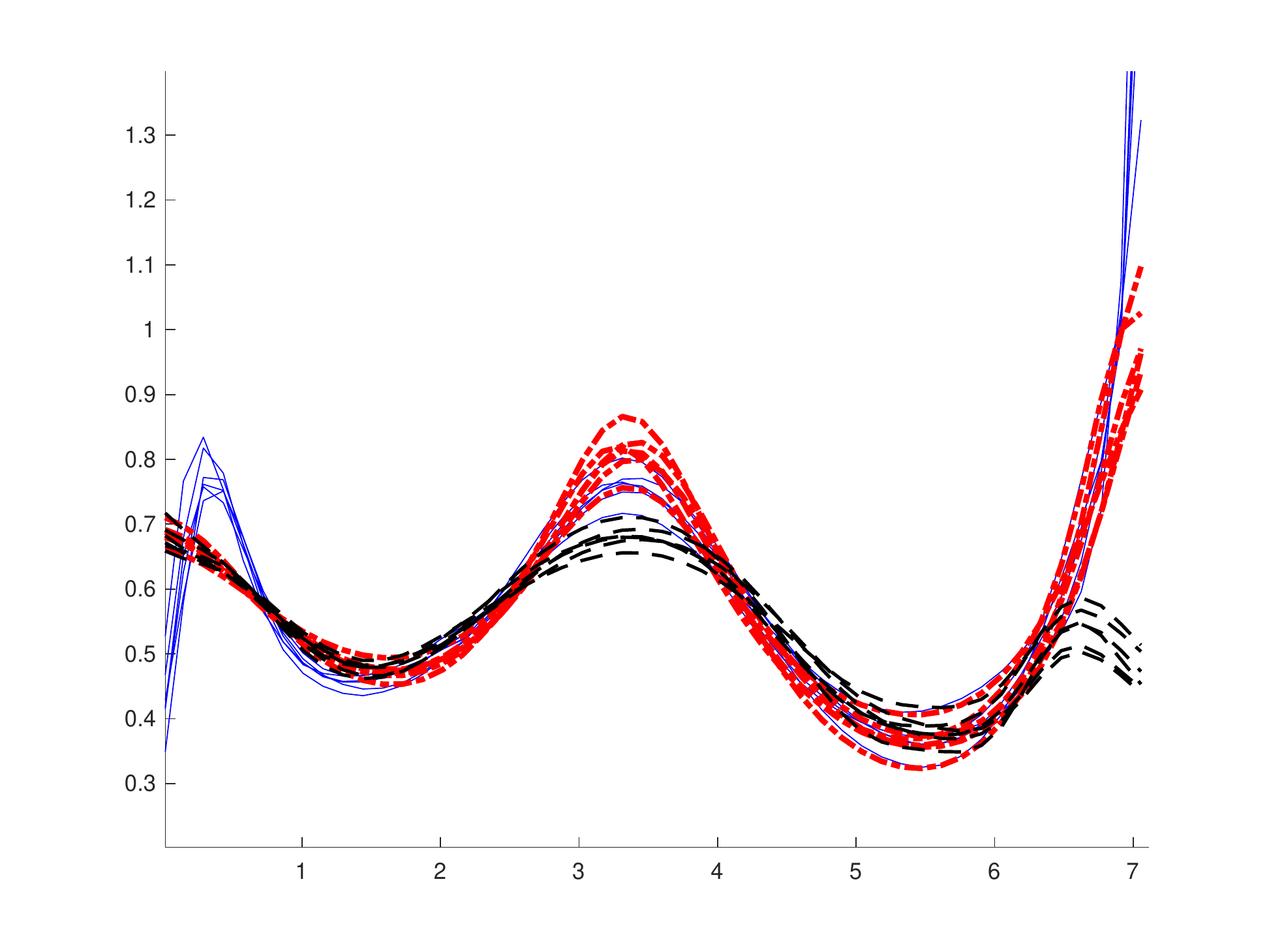} & \includegraphics[scale=0.35]{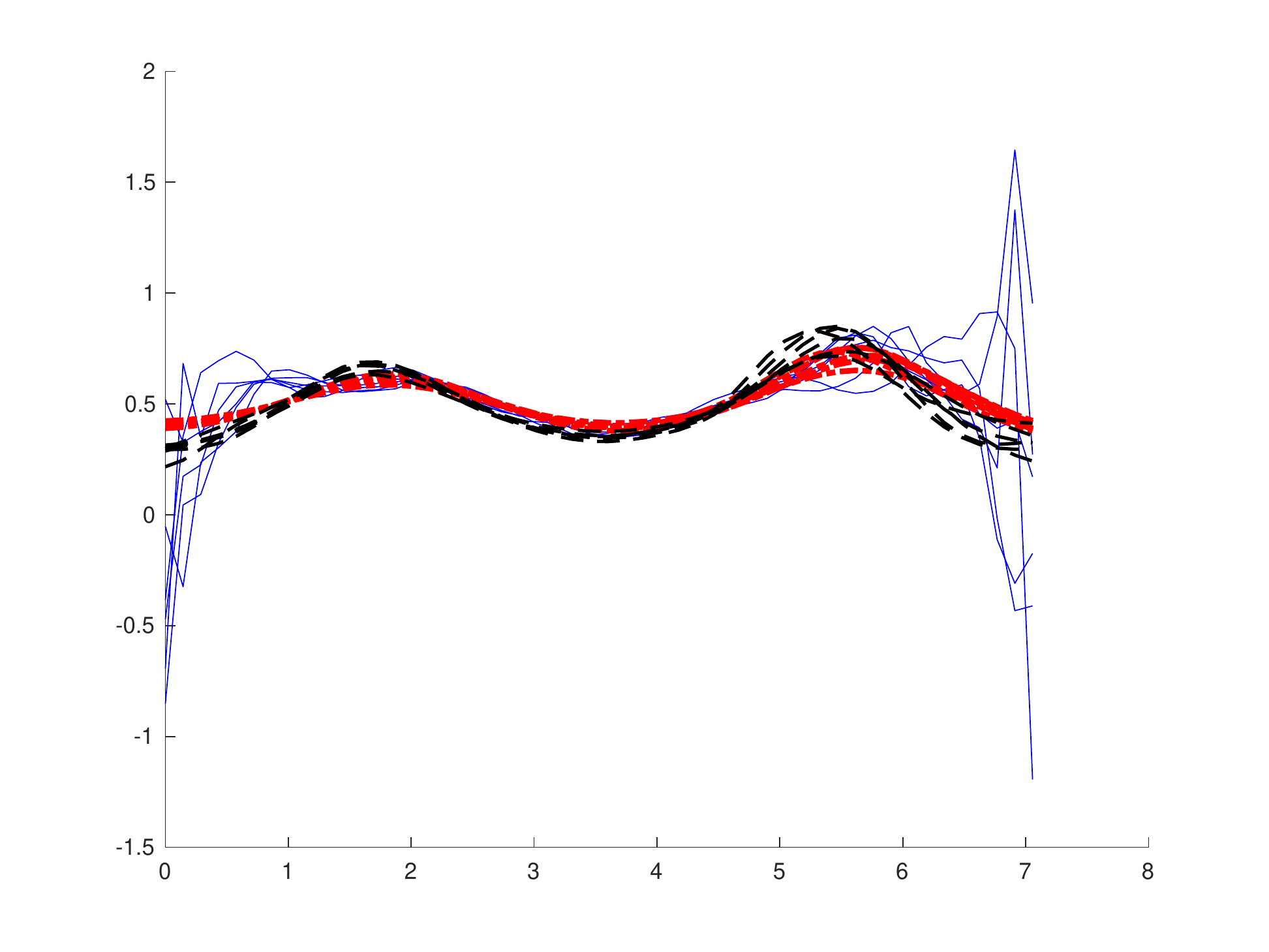} \\ 
\includegraphics[scale=0.35]{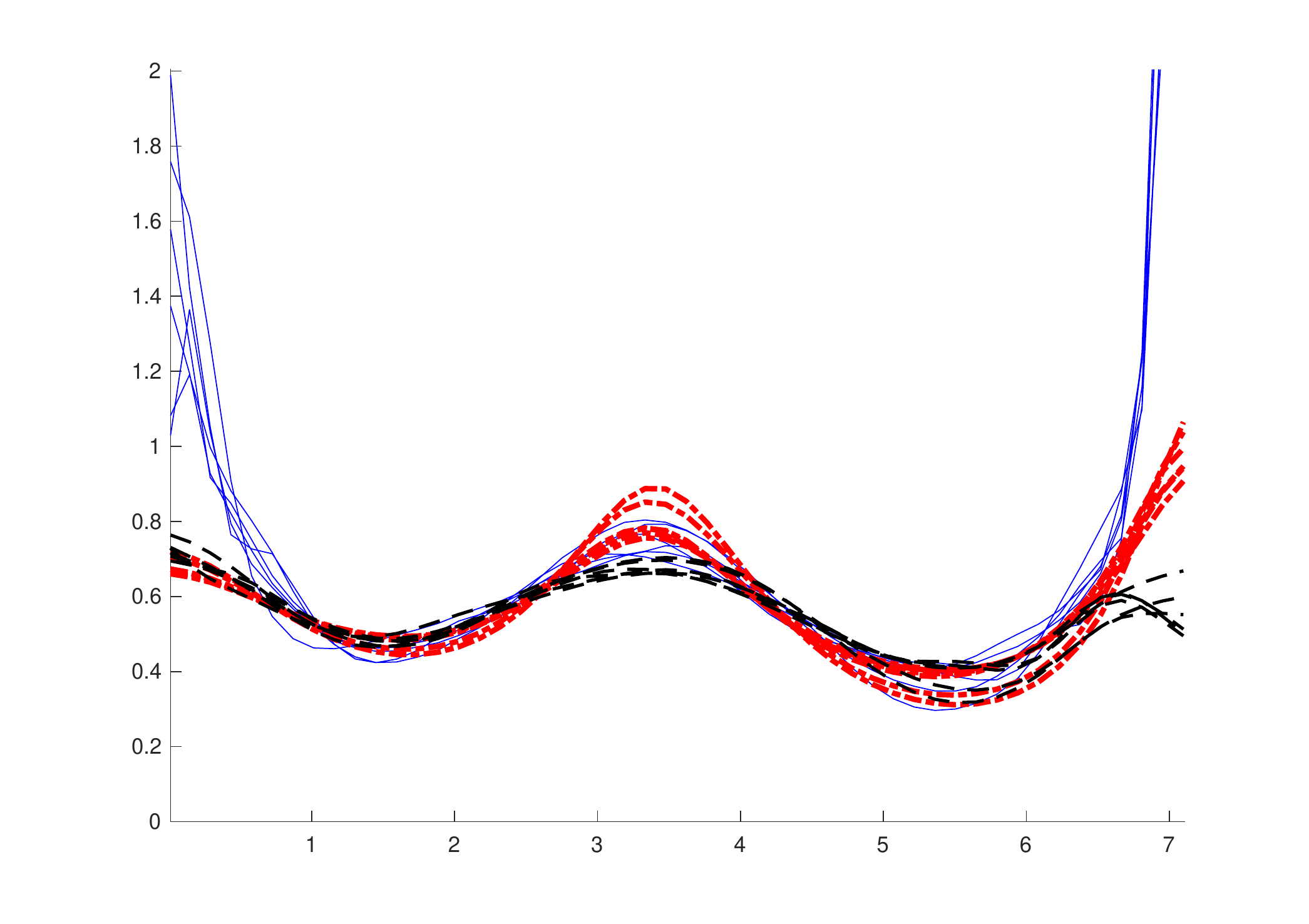} & \includegraphics[scale=0.35]{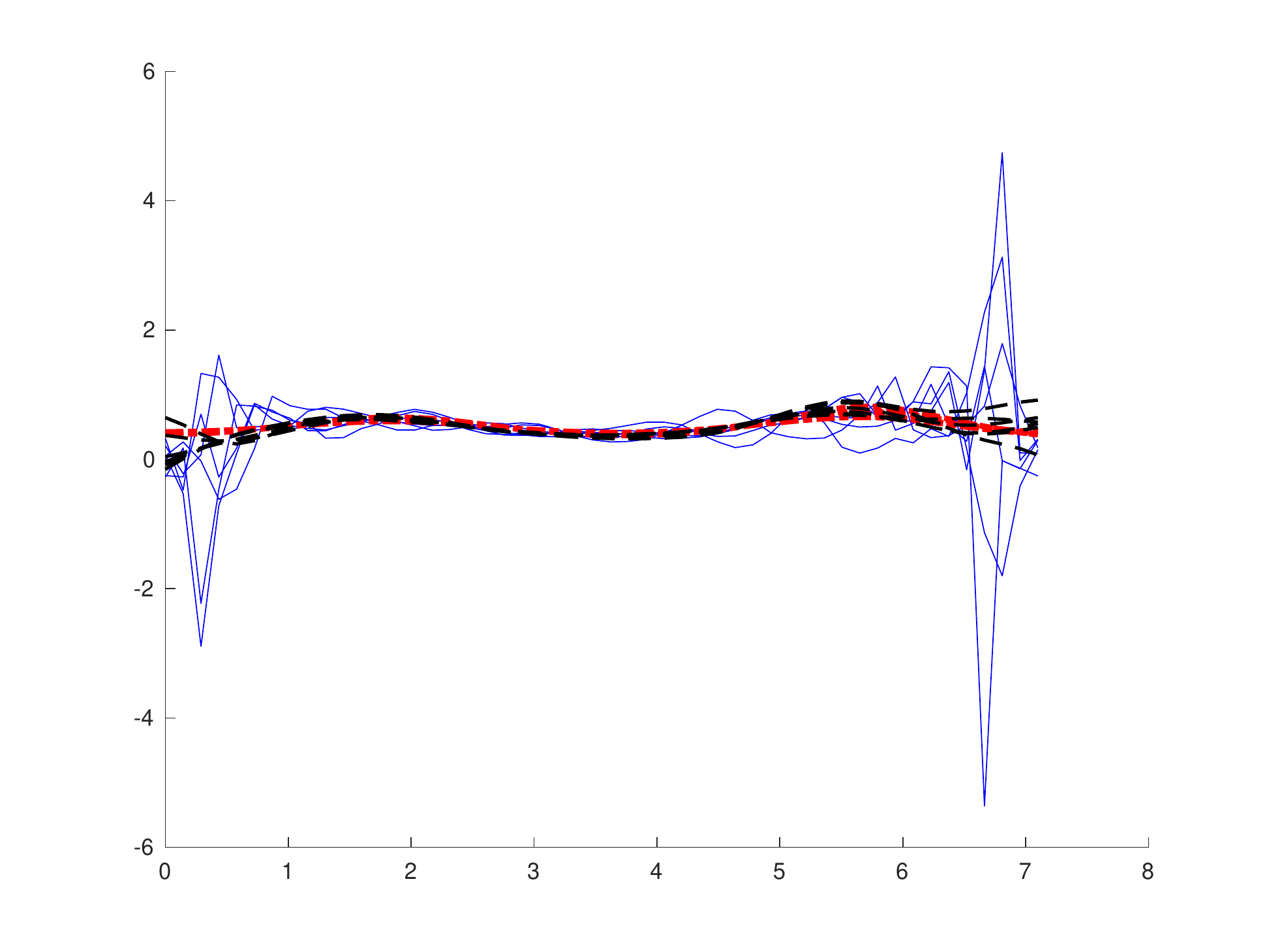} \\
\end{tabular}
\caption{Comparison of parameter estimates, $\hat\theta^{ind}_{Ext}$: mean of individual extrinsic estimates (solid blue), $\hat\theta^{pop}$: Frenet-Serret mean (dashed black), $\bar\theta$: mean generalised curvatures (dotted red) (top: $\sigma_{P}^{2}=0.04,\sigma_{e}^{2}=0.01$, bottom: $\sigma_{P}^{2}=0.04,\sigma_{e}^{2}=0.04$). Curvature (left) and Torsion (right). \label{fig:LargeSmall-LargeHigh}}
\end{centering}
\end{figure}

\subsection{Real data example}

We demonstrate our methodology with a dataset from a study of human movement behaviour in \citet{RaketMarkussen2016}. An experiment is designed to require each participant to move a hand-held object to a target location while avoiding an obstacle. The trajectories of the (three-dimensional) arm movement of each participant are recorded under various experimental conditions, with an aim to characterize the commonality and variations.  Figure~\ref{fig:rawX} shows an example of raw trajectories from one condition recorded from 10 subjects, each repeating 10 times.
We pre-process the data to create an arclength parametrized data $X_i$ defined on $[0, L_i], i=1,\ldots,n$ and define $Z_i(s) = X_i(sL_i)/L_i, s\in [0,1]$ as a length-normalized curve.
\begin{figure}
\begin{centering}
\includegraphics[scale=0.5]{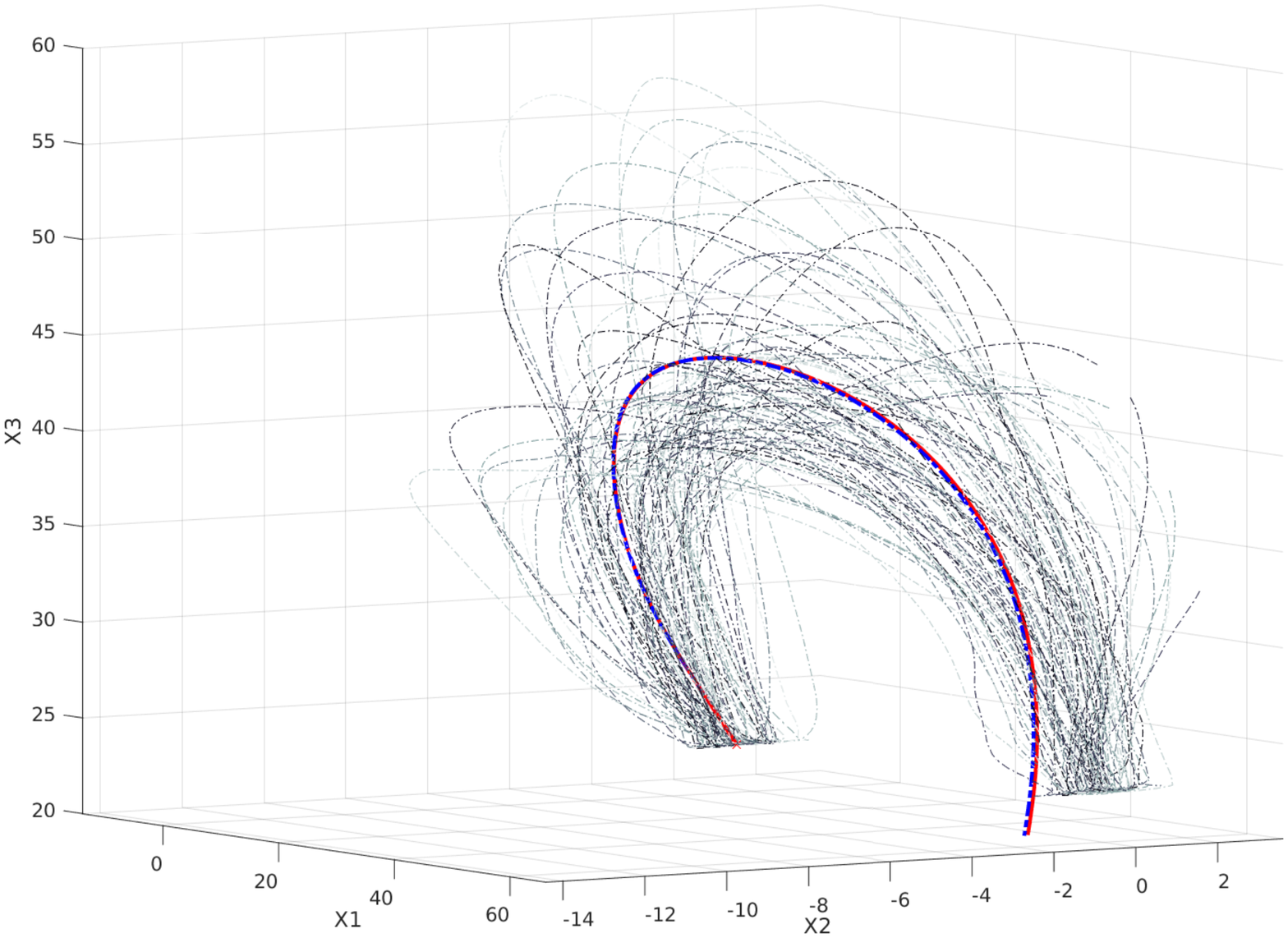} 
\caption{Raw data with 10 subjects with 10 replications, overlayed with the estimates of mean shape corresponding to parameters in Figure~\ref{fig:C1meanPar}. \label{fig:rawX}}
\end{centering}
\end{figure}
The noisy raw Frenet paths ($U_i$) are obtained from a constrained local polynomial smoothing with order 5 and a common bandwidth $h_1=0.2$ (chosen as a mean of the individual optimal one) on the normalized domain $[0,1]$. An example is shown in Figure~\ref{fig:rawQ}. For comparison, the mean Frenet paths corresponding to the parameter estimates in Figure~\ref{fig:C1meanPar} are also shown.

\begin{figure}
\begin{centering}
\includegraphics[scale=0.5]{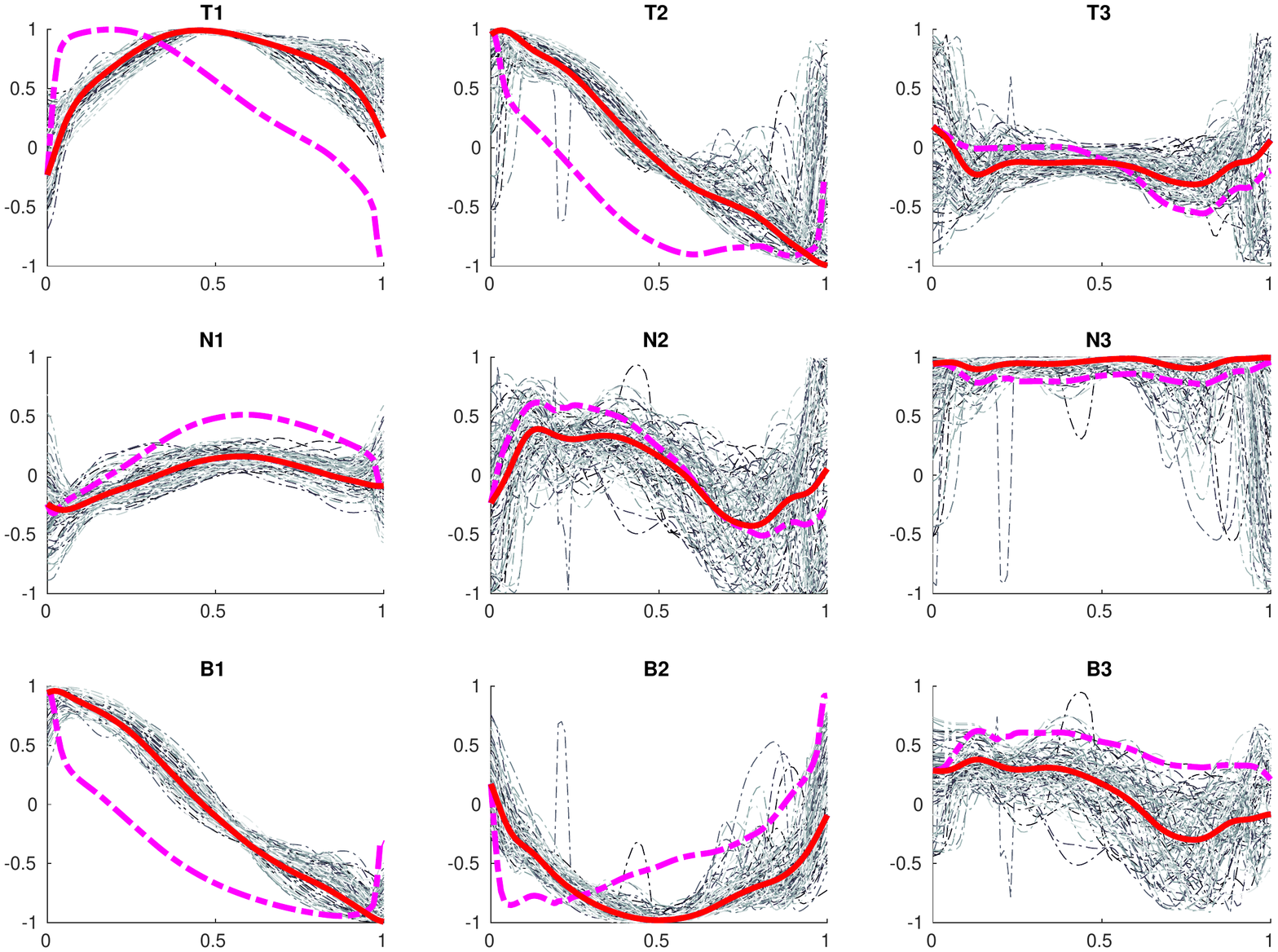} \vspace*{-1cm}
\caption{Noisy raw Frenet paths derived from the raw data in Figure~\ref{fig:rawX}. Overlayed are the mean Frenet paths corresponding to parameters in Figure~\ref{fig:C1meanPar}. \label{fig:rawQ}}
\end{centering}
\end{figure}

The estimates of the mean curvatures are shown in Figure~\ref{fig:C1meanPar}, under two smoothing parameter choices, $\lambda_{opt} = (10^{-8}, 10^{-7})$ by 10-fold cross validation and $\lambda_0=(0,0)$, thus resulting in comparable estimates. For comparison, the raw individual parameters estimated from extrinsic formulas using the local polynomial derivative estimates with its mean $(\hat\theta_{Ext})$. To mitigate the effect of some outliers, we also show its median estimate.  
The corresponding mean shape is shown in Figure~\ref{fig:rawX} with an average length. We do not include the estimates from extrinsic formula in the figure, as they were too far off. Since we do not assume the boundary conditions, there is more bias towards the end, and we did not attempt to remove the variation due to rotation. Nevertheless, the estimate seems a reasonable reflection of the average shape. Extending our methodology to compare multiple groups would be an interesting direction to explore for future work.  

\begin{figure}
\begin{centering}
\hspace*{-2cm}\includegraphics[scale=0.6]{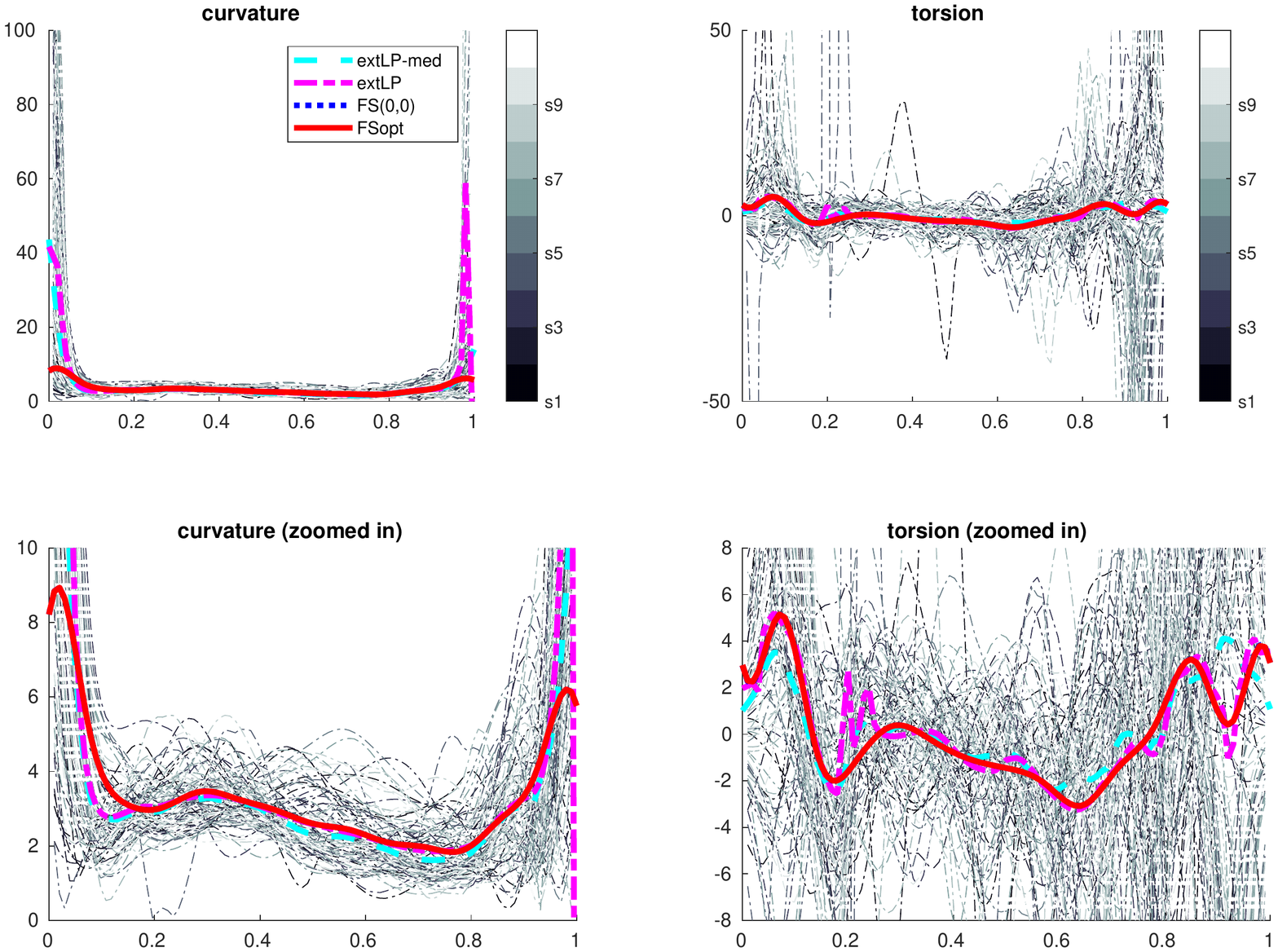} \vspace*{-1cm}
\caption{Mean curvature (left) and torsion (right) estimates with the noisy Frenet paths in Figure~\ref{fig:rawQ} using optimal parameters in thick (red) solid line (FSopt), fixed penalty parameters  (0,0) in think (blue) dotted line (FS(0,0)), extrinsic average in think (magenta) dashed line (extLP) and extrinsic median in think (cyan) dashed line (extLP.med). Overlayed are individual estimates from extrinsic formula (thin, dashed). Zoomed in views are given in the bottom. \label{fig:C1meanPar}}
\end{centering}
\end{figure}

\bibliography{Paper1_SerretFrenet}

\end{document}